\documentclass[sigplan,screen]{acmart}

\usepackage{amsmath}

\usepackage{amsthm}
\usepackage{amssymb}
\usepackage{multicol}
\usepackage{verbatim}
\usepackage{enumitem}
\usepackage{graphicx}
\usepackage{algorithm}
\usepackage{algpseudocode}
\usepackage{mathtools}
\usepackage{multirow}
\usepackage{bm}
\usepackage{nicefrac}
\usepackage{csquotes}

\usepackage{etoolbox}
\makeatletter
\patchcmd{\maketitle}{\@copyrightspace}{}{}{}
\makeatother
\usepackage{float}
\newfloat{algorithm}{t}{lop}

\usepackage{moresize}

\usepackage{color}
\usepackage{pifont}
\definecolor{mygreen}{RGB}{0,128,0}
\newcommand{\cmark}{\text{\color{mygreen}\ding{51}}}
\newcommand{\xmark}{\text{\color{red}\ding{55}}}
\definecolor{cerulean}{RGB}{0,123,167}
\newcommand{\textttb}[1]{\texttt{\color{cerulean}#1}}
\newcommand{\blue}[1]{{\color{cerulean}#1}}
\newcommand{\red}[1]{{\color{red}#1}}
\newcommand{\textttr}[1]{\texttt{\color{red}#1}}

\usepackage{xspace}
\newcommand{\toolname}{{\sc Atlas}\xspace}

\usepackage{tikz-cd}

\floatname{algorithm}{\small Algorithm}

\newlength{\pgmtab}          
\newenvironment{boxed-program}[1]{\begin{minipage}[#1]{9in}
\begin{tabbing}\hspace{\pgmtab}\=\hspace{\pgmtab}\=%
\hspace{\pgmtab}\=\hspace{\pgmtab}\=\hspace{\pgmtab}\=\hspace{\pgmtab}\=%
\hspace{\pgmtab}\=\hspace{\pgmtab}\=\hspace{\pgmtab}\=\hspace{\pgmtab}\=%
\hspace{\pgmtab}\=%
\kill}{\end{tabbing}\end{minipage}}

\newcommand{\TexComment}[1]{}

\newcommand{\New}{\text{New}}
\newcommand{\Assign}{\text{Assign}}
\newcommand{\Store}{\text{Store}}
\newcommand{\Load}{\text{Load}}
\newcommand{\FlowsTo}{\text{FlowsTo}}
\newcommand{\Alias}{\text{Alias}}

\newcommand{\Transfer}{\text{Transfer}}

\newcommand{\sss}{\hspace{0.1in}}

\makeatletter
\newcommand*{\da@rightarrow}{\mathchar"0\hexnumber@\symAMSa 4B }
\newcommand*{\da@leftarrow}{\mathchar"0\hexnumber@\symAMSa 4C }
\newcommand*{\xdashrightarrow}[2][]{%
  \mathrel{%
    \mathpalette{\da@xarrow{#1}{#2}{}\da@rightarrow{\,}{}}{}%
  }%
}
\newcommand{\xdashleftarrow}[2][]{%
  \mathrel{%
    \mathpalette{\da@xarrow{#1}{#2}\da@leftarrow{}{}{\,}}{}%
  }%
}
\newcommand*{\da@xarrow}[7]{%
  \sbox0{$\ifx#7\scriptstyle\scriptscriptstyle\else\scriptstyle\fi#5#1#6\m@th$}%
  \sbox2{$\ifx#7\scriptstyle\scriptscriptstyle\else\scriptstyle\fi#5#2#6\m@th$}%
  \sbox4{$#7\dabar@\m@th$}%
  \dimen@=\wd0 %
  \ifdim\wd2 >\dimen@
    \dimen@=\wd2 %
  \fi
  \count@=2 %
  \def\da@bars{\dabar@\dabar@}%
  \@whiledim\count@\wd4<\dimen@\do{%
    \advance\count@\@ne
    \expandafter\def\expandafter\da@bars\expandafter{%
      \da@bars
      \dabar@ 
    }%
  }%
  \mathrel{#3}%
  \mathrel{%
    \mathop{\da@bars}\limits
    \ifx\\#1\\%
    \else
      _{\copy0}%
    \fi
    \ifx\\#2\\%
    \else
      ^{\copy2}%
    \fi
  }%
  \mathrel{#4}%
}
\makeatother

\setcopyright{acmcopyright}
\begin{document}

\title{Active Learning of Points-To Specifications}


\author{Osbert Bastani}
\affiliation{
  \institution{Stanford University}            
  \country{USA}
}
\email{obastani@cs.stanford.edu}          

\author{Rahul Sharma}
\affiliation{
  \institution{Microsoft Research}            
  \country{India}
}
\email{rahsha@microsoft.com}          

\author{Alex Aiken}
\affiliation{
  \institution{Stanford University}            
  \country{USA}
}
\email{aiken@cs.stanford.edu}          

\author{Percy Liang}
\affiliation{
  \institution{Stanford University}            
  \country{USA}
}
\email{pliang@cs.stanford.edu}          



\begin{abstract}
When analyzing programs, large libraries pose significant challenges to static points-to analysis. A popular solution is to have a human analyst provide \emph{points-to specifications} that summarize relevant behaviors of library code, which can substantially improve precision and handle missing code such as native code. We propose \toolname, a tool that automatically infers points-to specifications. \toolname synthesizes unit tests that exercise the library code, and then infers points-to specifications based on observations from these executions. \toolname automatically infers specifications for the Java standard library, and produces better results for a client static information flow analysis on a benchmark of 46 Android apps compared to using existing handwritten specifications.
\end{abstract}



\maketitle

\section{Introduction}

When analyzing programs, dependencies on large libraries can significantly reduce the effectiveness of static analysis, since libraries frequently contain (i) native code that cannot be analyzed, (ii) use of challenging language features such as reflection, and (iii) deep call hierarchies that reduce precision. For example, the implementation of the \texttt{Vector} class in OpenJDK 1.7 uses multiple levels of call indirection and calls the native function \texttt{System.arrayCopy}.

A standard workaround is to use \emph{specifications} that summarize the relevant behaviors of library functions~\cite{zhu2013automated,facebook2017}. For a one-time cost of writing library specifications, the precision and soundness of the static analysis can improve dramatically when analyzing any client code. However, writing specifications for the entire library can be expensive~\cite{zhu2013automated,bastani2015specification}---many libraries contain a large number of functions, manually written specifications are often error prone~\cite{heule2016stratified}, and specifications must be updated whenever the library is updated.

To address these issues, approaches have been proposed for automatically \emph{inferring} specifications for library code, both based on dynamic analysis~\cite{nimmer2002automatic,alur2005synthesis,sharma2012interpolants,sharma2014invariant,bastani2015interactively} and on static analysis~\cite{ammons2002mining,kremenek2006uncertainty,ramanathan2007static,shoham2008static,livshits2009merlin,beckman2011probabilistic}. In particular, tools have been designed to infer properties of missing code, including taint flow properties~\cite{clapp2015modelgen}, function models~\cite{heule2015mimic,heule2016stratified}, and callback control flow~\cite{jeon2016synthesizing}. While these approaches are incomplete, and may not infer sound specifications, current static analyses used in production already rely on user-provided specifications~\cite{facebook2017}, and as we will show, tools that automatically infer specifications can outperform human analysts.

We propose an algorithm based on dynamic analysis that infers library specifications summarizing points-to effects relevant to a flow-insensitive points-to analysis. Our algorithm works by iteratively guessing candidate specifications and then checking whether each one is ``correct''. Intuitively, a specification is correct if it must be included in any sound set of specifications, i.e., it is precise. This property ensures that the specification does not cause any false positives. There are two constraints that make our problem substantially more challenging than previously studied settings:
\begin{itemize}
\item Points-to effects cannot be summarized for a library function in isolation, e.g., in Figure~\ref{fig:example}, \texttt{set}, \texttt{get}, and \texttt{clone} all refer to the same field \texttt{f}. Thus, a guessed candidate must simultaneously summarize the points-to effects of \texttt{set}, \texttt{get}, and \texttt{clone}. Furthermore, the inference algorithm may not know the library field \texttt{f} exists, and must invent a \emph{ghost field} to represent it.
\item It is often difficult to instrument library code (e.g., native code); we assume only blackbox access.
\end{itemize}

To address these challenges, we introduce \emph{path specifications} to describe points-to effects of library code. Each path specification summarizes a single points-to effect of a combination of functions. An example is:
\begin{displayquote}
For two calls \texttt{box.set(x)} and \texttt{box.get()}, the return value of \texttt{get} may alias \texttt{x}.
\end{displayquote}
Path specifications have two desirable properties: (i) we can check if a candidate path specification is correct using input-output examples, and (ii) correctness of path specifications is independent, i.e., a set of path specifications is correct if each specification in the set is correct. These two properties imply that we can infer path specifications incrementally. In particular, we propose an active learning algorithm that infers specifications by actively identifying promising candidate specifications, and then automatically checking the correctness of each candidate independently.

\begin{figure*}
\begin{minipage}{0.45\textwidth}
\footnotesize
\begin{verbatim}
class Box { // library
  Object f;
  void set(Object ob) { f = ob; }
  Object get() { return f; }
  Box clone() {
    Box b = new Box(); // o_clone
    b.f = f;
    return b; } }
\end{verbatim}
\end{minipage}
\begin{minipage}{0.45\textwidth}
\footnotesize
\begin{verbatim}
boolean test() { // program
  Object in = new Object(); // o_in
  Box box = new Box(); // o_box
  box.set(in);
  Object out = box.get();
  return in == out; }
\end{verbatim}
\end{minipage}
\caption{Implementation of the library methods \texttt{set}, \texttt{get}, and \texttt{clone} in the \texttt{Box} class (left), and an example of a program using these functions (right).}
\label{fig:example}
\end{figure*}

We implement our algorithm in a tool called \toolname,\footnote{\toolname stands for AcTive Learning of Alias Specifications, and is available at \url{https://github.com/obastani/atlas}.} which infers path specifications for functions in Java libraries. In particular, we evaluate \toolname by using it to infer specifications for the Java standard library, including the Collections API, which contains many functions that exhibit complex points-to effects. We evaluate \toolname using a standard static explicit information flow client~\cite{fuchs2009scandroid,feng2014apposcopy,arzt2014flowdroid,bastani2015interactively} on a benchmark of 46 Android apps. Our client uses handwritten points-to specifications for the Java standard library~\cite{zhu2013automated,bastani2015specification}. Over the course of two years, we have handwritten a number of points-to specifications tailored to finding information flows in apps we analyzed, including those in our benchmark. However, these specifications remain incomplete.

We use \toolname to infer specifications for classes in six commonly used packages in the Java standard library. \toolname inferred more than 5$\times$ as many specifications as the existing, handwritten ones. We show that using the specifications inferred by \toolname instead of the handwritten ones improves the recall of our information flow client by 52\%. Moreover, we manually wrote ground truth specifications for the most frequently used subset of classes in the Collections API, totalling 1,731 lines of code, and show that \toolname inferred the correct specifications (i.e., identical to handwritten ground truth specifications) for 97\% of these functions. Finally, we show that on average, using specifications more than halves the number of false positive points-to edges compared to directly analyzing the library implementation---in particular, the library implementation contains deep call hierarchies, which our points-to analysis is insufficiently context-sensitive to handle. 
Our contributions are:
\begin{itemize}
\item We propose path specifications and prove they are sufficiently expressive to precisely model library functions for a standard flow-insensitive points-to analysis.
\item We formulate the problem of inferring path specifications as a language inference problem and  design a novel specification inference algorithm. 
\item We implement our approach in \toolname, and use it to infer a large number of specifications for the Java standard library. We use these inferred specifications to automatically replicate and even improve the results of an existing static information flow client.
\end{itemize}

\section{Overview}
\label{sec:overview}

Consider the program \texttt{test} shown in Figure~\ref{fig:example}. We assume that the static analysis resolves aliasing using Andersen's analysis~\cite{andersen1994program}, a context- and flow-insensitive points-to analysis; our approach also applies to cloning-based context- and object-sensitive extensions~\cite{whaley2004cloning}. To determine that variables \texttt{in} and \texttt{out} may be aliased, the points-to analysis has to reason about the heap effects of the \texttt{set} and \texttt{get} methods in the \texttt{Box} class.
The analyst can write \emph{specifications} for library functions that summarize their heap effects with respect to the semantics of the points-to analysis. For example, consider the \texttt{Stack} class in the Java library: its implementation extends \texttt{Vector}, which is challenging to analyze due to deep call hierarchies and native code. The following are the specifications for the \texttt{Stack} class, implemented as code fragments that overapproximate the heap effects of the methods:
\begin{small}
\begin{verbatim}
class Stack { // specification
  Object f; // ghost field
  void push(Object ob) { f = ob; }
  Object pop() { return f; } }
\end{verbatim}
\end{small}
In contrast to the implementation, the specifications for \texttt{Stack} are simple and easy to analyze. In fact, when used with our static points-to analysis, these specifications are not only sound but also precise, since our points-to analysis is flow-insensitive and collapses arrays into a single field.

A typical approach is to write specifications lazily~\cite{zhu2013automated,bastani2015specification}---the analyst examines the program being analyzed, identifies the library functions most relevant to the static analysis, and writes specifications for them. The effort invested in writing specifications helps reduce the labor required to discharge false positives.\footnote{In contrast to the recurring cost of debugging false positives, the cost of implementing specifications is amortized over many programs.} The analyst can trade off between manual effort and soundness by expending more effort as needed to increase the completeness of the library specifications.

\toolname helps bootstrap the specification writing process by automatically inferring points-to specifications. In accordance with the goal of minimizing false positives, \toolname employs a two-phase approach that prioritizes the precision of the specifications it infers (where precision is defined with respect to our points-to analysis). In the first phase, \toolname infers specifications guaranteed to be precise. In particular, each inferred specification $s$ comes with a \emph{witness}, which intuitively is a unit test that proves the precision of $s$ by exhibiting the heap effects specified by $s$. In the second phase, \toolname inductively generalizes the specifications inferred in the first phase, using a large number of checks to minimize the possibility of imprecision.

\paragraph{{\bf\em Specification search space.}}

The main challenge in inferring points-to specifications is to formulate the search space of candidate specifications. Na\"{i}vely searching over the space of code fragments is ineffective because specifications are highly interdependent, i.e., the specification of a library function is dependent on the specifications of other functions. For example, the specifications for each of the methods in the \texttt{Box} class all refer to the shared field \texttt{f}. Thus, to infer precise specifications, the na\"{i}ve algorithm would have to guess a code fragment for every function in the library (or at least in a class) before it can check for precision; the likelihood that it could make a correct guess is tiny.

Our key insight is that while specifications cannot be broken at function boundaries, we can decompose them in a different way. In particular, we propose \emph{path specifications}, which are independent units of points-to effects. Intuitively, a path specification is a dataflow path an object might take through library code. An example of a path specification is
\begin{align}
\label{eqn:expathspec}
s_{\text{box}}=\texttt{ob}\dashrightarrow\texttt{this}_{\texttt{set}}\rightarrow\texttt{this}_{\texttt{get}}\dashrightarrow r_{\texttt{get}}.
\end{align}
This path specification roughly has the following semantics:
\begin{itemize}
\item Suppose that an object $o$ enters the library code as the parameter \texttt{ob} of \texttt{set}; then, it is associated with the receiver $\texttt{this}_{\texttt{set}}$ of the \texttt{set} method (the edge $\texttt{ob}\dashrightarrow\texttt{this}_{\texttt{set}}$).
\item Suppose that in the client program, $\texttt{this}_{\texttt{set}}$ is aliased with $\texttt{this}_{\texttt{get}}$ (the edge $\texttt{this}_{\texttt{set}}\rightarrow\texttt{this}_{\texttt{get}}$); then, $o$ enters the \texttt{get} method.
\item Then, $o$ exits the library as the return value $r_{\texttt{get}}$ of \texttt{get}.
\end{itemize}
In particular, the dashed edges $\texttt{ob}\dashrightarrow\texttt{this}_{\texttt{set}}$ and $\texttt{this}_{\texttt{get}}\dashrightarrow r_{\texttt{get}}$ represent the effects of library code, and the solid edge $\texttt{this}_{\texttt{set}}\rightarrow\texttt{this}_{\texttt{get}}$ represents an assumption about the effects of client code. Then, the path specification says that, if the points-to analysis determines that $\texttt{this}_{\texttt{get}}$ and $\texttt{this}_{\texttt{set}}$ are aliased, then \texttt{ob} is \emph{transfered} to $r_{\texttt{get}}$, which intuitively means that \texttt{ob} is indirectly assigned to $r_{\texttt{get}}$. More precisely, the semantics of $s_{\text{box}}$ is the following logical formula:
\begin{align}
\label{eqn:exlogicspec}
(\texttt{this}_{\texttt{set}}\xrightarrow{\Alias}\texttt{this}_{\texttt{get}})\Rightarrow(\texttt{ob}\xrightarrow{\Transfer}r_{\texttt{get}}).
\end{align}
Here, an edge $x\xrightarrow{A}y$ says that program variables $x$ and $y$ satisfy a binary relation $A$. We describe path specifications in detail in Section~\ref{sec:path}.

\paragraph{{\bf\em Testing specifications.}}

Path specifications satisfy two key requirements. The first requirement is that we can check the precision of a single path specification. For example, for the specification $s_{\text{box}}$, consider the program \texttt{test} shown in Figure~\ref{fig:example}. This unit test satisfies three properties:
\begin{itemize}
\item It satisfies the antecedent of $s_{\text{box}}$, since $\texttt{this}_{\texttt{set}}$ and $\texttt{this}_{\texttt{get}}$ are aliased.
\item It does not induce any other relations between the variables at the interface of the \texttt{Box} class (i.e., \texttt{ob}, $\texttt{this}_{\texttt{set}}$, \texttt{i}, $\texttt{this}_{\texttt{get}}$, and $r_{\texttt{get}}$).
\item It returns the consequent of $s_{\text{box}}$, i.e., whether \texttt{ob} and $r_{\texttt{get}}$ point to the same object.
\end{itemize}
Upon executing \texttt{test}, it returns that the consequent of $s_{\text{box}}$ is true; therefore, any sound set of specifications must include $s_{\text{box}}$, so $s_{\text{box}}$ is precise (we formalize soundness and precision of path specifications in Section~\ref{sec:path}). We refer to such a unit test as a \emph{witness} for the path specification. In summary, as long as we can find a witness for a candidate path specification $s$, then we can guarantee that $s$ is precise. One caveat is that even if \texttt{test} returns false, $s_{\text{box}}$ may still be precise.

Therefore, to check if a candidate path specification $s$ is precise, our algorithm can synthesize a unit test $P$ similar to \texttt{test} and execute it. If $P$ returns true, then our algorithm concludes that $s$ is precise; otherwise, it (conservatively) concludes that $s$ is imprecise. Note that $s$ may be precise even if $P$ returns false; this possibility is unavoidable since executions are underapproximations, i.e., $P$ may not exercise all paths of the relevant library functions. We use heuristics to ensure that our algorithm rarely concludes that precise specifications are imprecise. We describe our unit test synthesis algorithm in detail in Section~\ref{sec:testsynthesis}.

The second requirement of path specifications is that they are independent, i.e., given a set of path specifications for which each specification has a witness, then the set as a whole is precise in the sense that it is a subset of any sound set of path specifications. In other words, we can use a potential witness to check the precision of a path specification in isolation. Thus, our specification inference algorithm can focus on inferring individual path specifications.

\paragraph{{\bf\em Specification inference.}}

Our specification inference algorithm uses two inputs:
\begin{itemize}
\item {\bf Library interface:} The type signature of each function in the library.
\item {\bf Blackbox access:} The ability to execute sequences of library functions on chosen inputs and obtain the corresponding outputs.
\end{itemize}
In the first phase (described in Section~\ref{sec:sample}), \toolname randomly guesses a candidate path specification $s$, synthesizes a potential witness for $s$, and retains $s$ if this unit test returns true. This process is repeated a large number of times to obtain a large but finite set of precise path specifications.

To soundly model the library, an infinite number of path specifications may be required, e.g., the path specifications required to soundly model \texttt{set}, \texttt{get}, and \texttt{clone} are
\begin{align}
\label{eqn:exinfinite}
\texttt{ob}\dashrightarrow\texttt{this}_{\texttt{set}}\blue{(}&\to\texttt{this}_{\texttt{clone}}\dashrightarrow r_{\texttt{clone}}\blue{)^*} \\
&\to\texttt{this}_{\texttt{get}}\dashrightarrow r_{\texttt{get}}. \nonumber
\end{align}
These specifications say that if we call \texttt{set}, then call \texttt{clone} $n$ times in sequence, and finally call \texttt{get} (all with the specified aliasing between receivers and return values), then the parameter \texttt{ob} of \texttt{set} is transfered to the return value of \texttt{get}.

\begin{figure}
\centering
\footnotesize
\[\text{(assign)}~\dfrac{y\gets x}{x\xrightarrow{\Assign}y}\hspace{0.2in}\text{(allocation)}~\dfrac{o=(x\gets X())}{o\xrightarrow{\New}o}\]
\[\text{(store)}~\dfrac{y.f\gets x}{x\xrightarrow{\Store[f]}y}\hspace{0.2in}\text{(load)}~\dfrac{y\gets x.f}{x\xrightarrow{\Load[f]}y}\hspace{0.2in}\text{(backwards)}~\dfrac{x\xrightarrow{\sigma}y}{y\xrightarrow{\overline{\sigma}}x}\] \\
\[\text{(call parameter)}~\dfrac{y\gets m(x)}{x\xrightarrow{\Assign}p_m}\hspace{0.2in}\text{(call return)}~\dfrac{y\gets m(x)}{r_m\xrightarrow{\Assign}y}\]
\caption{Rules for constructing a graph $G$ encoding the relevant semantics of program statements.}
\label{fig:constructgraph}
\end{figure}

Thus, in the second phase (described in Section~\ref{sec:rpni}), we inductively generalize the set $S$ of path specifications in the first phase to a possibly infinite set. We leverage the idea that a path specification can be represented as a sequence of variables $s\in\mathcal{V}_{\text{path}}^*$, where $\mathcal{V}_{\text{path}}$ are the variables in the library interface---for example, $s_{\text{box}}$ corresponds to
\begin{align*}
  \texttt{ob}~\texttt{this}_{\texttt{set}}~\texttt{this}_{\texttt{get}}~r_{\texttt{get}}\in\mathcal{V}_{\text{path}}^*.
\end{align*}
Thus, a set of path specifications is a formal language over the alphabet $\mathcal{V}_{\text{path}}$. As a consequence, we can frame the inductive generalization problem as a language inference problem: given (i) the finite set of positive examples from phase one, and (ii) an oracle we can query to determine whether a given path specification $s$ is precise (though this oracle is \emph{noisy}, i.e., it may return false even if $s$ is precise), the goal is to infer a (possibly infinite) language $S\subseteq\mathcal{V}_{\text{path}}^*$.

We devise an active language learning algorithm to solve this problem. Our algorithm  proposes candidate inductive generalizations of $S$, and then checks the precision of these candidates using a large number of synthesized unit tests. Unlike phase one, we may introduce imprecision even if all the unit tests pass; we show empirically that precision is maintained. Finally, our algorithm infers a regular set of path specifications. In theory, no regular set may suffice to model the library code and more expressive language inference techniques might be required~\cite{bastani2017synthesizing}. In practice, we find that regular sets are sufficient.

In our example, suppose that phase one infers
\begin{align*}
  \texttt{ob}\dashrightarrow\texttt{this}_{\texttt{set}}&\to\texttt{this}_{\texttt{clone}}\dashrightarrow r_{\texttt{clone}} \\
  &\to\texttt{this}_{\texttt{clone}}\dashrightarrow r_{\texttt{clone}}\to\texttt{this}_{\texttt{get}}\dashrightarrow r_{\texttt{get}}.
\end{align*}
Then, phase two would inductively generalize it to (\ref{eqn:exinfinite}). Finally, our tool automatically converts these path specifications to the equivalent code fragment specifications shown in Figure~\ref{fig:example}. These code fragment specifications can be used in place of the (possibly unavailable) library implementation when analyzing client code.

\section{Background on Points-To Analysis}
\label{sec:background}

We consider programs with assignments $y\gets x$ (where $x,y\in\mathcal{V}$ are variables), allocations $x\gets X()$ (where $X\in\mathcal{C}$ is a type), stores $y.f\gets x$ and loads $y\gets x.f$ (where $f\in\mathcal{F}$ is a field), and calls to library functions $y\gets m(x)$ (where $m\in\mathcal{M}$ is a library function). For simplicity, we assume that each library function $m$ has a single parameter $p_m$ and a return value $r_m$.

An \emph{abstract object} $o\in\mathcal{O}$ is an allocation statement $o=(x\gets X())$. A \emph{points-to edge} is a pair $x\hookrightarrow o\in\mathcal{V}\times\mathcal{O}$. A static points-to analysis computes points-to edges $\Pi\subseteq\mathcal{V}\times\mathcal{O}$. Our results are for Andersen's analysis, a flow-insensitive points-to analysis~\cite{andersen1994program}, but generalize to object- and context-sensitive extensions based on cloning~\cite{whaley2004cloning}. We describe the formulation of Andersen's analysis as a context-free language reachability problem~\cite{reps1998program,kodumal2004set,kodumal2005banshee,sridharan2005demand,sridharan2006refinement}.

\begin{figure}
\footnotesize
\begin{align*}
\Transfer&\to\varepsilon\mid\Transfer~\Assign\mid\Transfer~\Store[f]~\Alias~\Load[f] \\
\overline{\Transfer}&\to\varepsilon\mid\overline{\Assign}~\overline{\Transfer}\mid\overline{\Load[f]}~\Alias~\overline{\Store[f]}~\overline{\Transfer} \\
\Alias&\to\overline{\Transfer}~\overline{\New}~\New~\Transfer \\
\FlowsTo&\to\New~\Transfer
\end{align*}
\caption{Productions for the context-free grammar $C_{\text{pt}}$. The start symbol of $C_{\text{pt}}$ is $\FlowsTo$.}
\label{fig:pointstogrammar}
\end{figure}

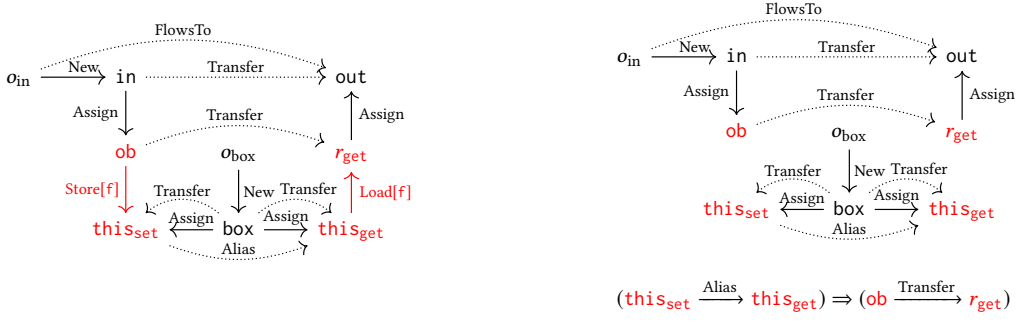
\begin{figure*}
\footnotesize
\begin{minipage}{0.45\textwidth}
\[
\begin{tikzcd}[column sep=2.0em,row sep=2.0em]
o_{\text{in}} \arrow{r}[xshift=1.2ex]{\New} \arrow[densely dotted,bend left=20]{rrr}{\FlowsTo}
& \texttt{in} \arrow{d}[xshift=-6.4ex]{\Assign} \arrow[densely dotted]{rr}{\Transfer}
&
& \texttt{out}
\\
& \textttr{ob} \arrow[red]{d}[xshift=-7.2ex]{\red{\Store[\texttt{f}]}} \arrow[densely dotted,bend left=20]{rr}{\Transfer}
& o_{\text{box}} \arrow[d,"{\New}"]
& \red{r_{\texttt{get}}} \arrow{u}[xshift=6.4ex]{\Assign}
\\
& \textttr{this}_{\textttr{set}} \arrow[densely dotted,bend right=20]{rr}{\Alias}
& \texttt{box} \arrow{l}[yshift=2.5ex]{\Assign} \arrow[r,"{\Assign}"]
\arrow[densely dotted,bend right=40]{l}[yshift=2.2ex]{\Transfer} \arrow[densely dotted,bend left=40]{r}[xshift=1.5ex]{\Transfer}
& \textttr{this}_{\textttr{get}} \arrow[red]{u}[xshift=7.2ex]{\red{\Load[\texttt{f}]}} \\
\end{tikzcd}
\]
\end{minipage}
\begin{minipage}{0.45\textwidth}
\[
\begin{tikzcd}[column sep=2.0em,row sep=2.0em]
o_{\text{in}} \arrow{r}[xshift=1.2ex]{\New} \arrow[densely dotted,bend left=20]{rrr}{\FlowsTo}
& \texttt{in} \arrow{d}[xshift=-6.4ex]{\Assign} \arrow[densely dotted]{rr}{\Transfer}
&
& \texttt{out}
\\
& \textttr{ob} \arrow[densely dotted,bend left=20]{rr}{\Transfer}
& o_{\text{box}} \arrow[d,"{\New}"]
& \red{r_{\texttt{get}}} \arrow{u}[xshift=6.4ex]{\Assign}
\\
& \textttr{this}_{\textttr{set}} \arrow[densely dotted,bend right=20]{rr}{\Alias}
& \texttt{box} \arrow{l}[yshift=2.5ex]{\Assign} \arrow[r,"{\Assign}"]
\arrow[densely dotted,bend right=40]{l}[yshift=2.2ex]{\Transfer} \arrow[densely dotted,bend left=40]{r}[xshift=1.5ex]{\Transfer}
& \textttr{this}_{\textttr{get}} \\
\end{tikzcd}
\]
\[(\red{\texttt{this}_{\texttt{set}}}\xrightarrow{\Alias}\red{\texttt{this}_{\texttt{get}}})\Rightarrow(\red{\texttt{ob}}\xrightarrow{\Transfer}\red{r_{\texttt{get}}})\]
\end{minipage}
\caption{The left-hand side shows the graph $\widetilde{G}$ computed by the static analysis with the library code, and the right-hand side shows $\widetilde{G}$ computed with path specifications (the relevant path specification is shown below the graph). The solid edges are the graph $G$ extracted from the program \texttt{test} shown in Figure~\ref{fig:example}. In addition, the dashed edges are a few of the edges in $\widetilde{G}$ when computing the transitive closure. We omit backward edges (i.e., with labels $\overline{A}$) for clarity. Vertices and edges corresponding to library code are highlighted in red.}
\label{fig:graph}
\end{figure*}

\paragraph{{\bf\em Graph representation.}}

First, our static analysis constructs a labeled graph $G$ representing the program semantics. The vertices of $G$ are $\mathcal{V}\cup\mathcal{O}$. The edge labels
\begin{align*}
\Sigma_{\text{pt}}=\{\Assign,\New,\Store,\Load,\overline{\Assign},\overline{\New},\overline{\Load},\overline{\Store}\}
\end{align*}
encode the semantics of program statements. The rules for constructing $G$ are in Figure~\ref{fig:constructgraph}. For example, the edges extracted for the program \texttt{test} in Figure~\ref{fig:example} are the solid edges in Figure~\ref{fig:graph}.

\paragraph{{\bf\em Transitive closure.}}

Second, our static analysis computes the \emph{transitive closure} $\widetilde{G}$ of $G$ according to the context-free grammar $C_{\text{pt}}$ in Figure~\ref{fig:pointstogrammar}. A \emph{path} $y\xdashrightarrow{\alpha}x$ in $G$ is a sequence
\begin{align*}
x\xrightarrow{\sigma_1}v_1\xrightarrow{\sigma_2}...\xrightarrow{\sigma_k}y
\end{align*}
of edges in $G$ such that $\alpha=\sigma_1...\sigma_k\in\Sigma_{\text{pt}}^*$. Then, $\widetilde{G}$ contains (i) each edge $x\xrightarrow{\sigma}y$ in the original graph $G$, and (ii) if there is a path $x\xdashrightarrow{\alpha}y$ in $G$ such that $A\xRightarrow{*}\alpha$ (where $A$ is a nonterminal of $C_{\text{pt}}$), the edge $x\xrightarrow{A}y$. Our static analysis computes $\widetilde{G}$ using dynamic programming; e.g., see~\cite{melski2000interconvertibility}.

The first production in Figure~\ref{fig:pointstogrammar} constructs the \emph{transfer} relation $x\xrightarrow{\Transfer}y$, which says that $x$ may be ``indirectly assigned'' to $y$. The second production constructs the ``backwards'' transfer relation. The third production constructs the \emph{alias} relation $x\xrightarrow{\Alias}y$, which says that $x$ may alias $y$. The fourth production computes the points-to relation, i.e., $x\hookrightarrow o$ whenever $o\xrightarrow{\FlowsTo}x\in\widetilde{G}$.

\section{Path Specifications}
\label{sec:path}

At a high level, a path specification encodes when edges in $G$ would have been connected by (missing) edges from the library implementation. In particular, suppose that our static analysis could analyze the library implementation, and that while computing the transitive closure $\widetilde{G}$, it contains a path
\begin{align}
\label{eqn:ppath}
z_1\xdashrightarrow{\beta_1}w_1\xrightarrow{A_1}z_2\xdashrightarrow{\beta_2}...\xrightarrow{A_{k-1}}z_k\xdashrightarrow{\beta_k}w_k.
\end{align}
Here, $z_1,w_1,...,z_k,w_k\in\mathcal{V}_{\text{path}}$ where $\mathcal{V}_{\text{path}}$ is the set of variables in the library interface (i.e., parameters and return values of library functions), $\beta_1,...,\beta_k\in\Sigma_{\text{pt}}^*$ are labels for paths corresponding to library code, and $A_1,...,A_{k-1}$ are nonterminals in $C_{\text{pt}}$ labeling relations computed so far. If
\begin{align*}
A\xRightarrow{*}\beta_1A_1...\beta_{k-1}A_{k-1}\beta_k,
\end{align*}
then our analysis adds $z_1\xrightarrow{A}w_k$ to $\widetilde{G}$. Path specifications ensure that our analysis adds such edges to $\widetilde{G}$ when the library code is unavailable (so the paths $z_i\xdashrightarrow{\beta_i}w_i$ are missing from $\widetilde{G}$); e.g., a path specification for (\ref{eqn:ppath}) says that if
\begin{align*}
w_1\xrightarrow{A_1}z_2,~...,~w_{k-1}\xrightarrow{A_{k-1}}z_k\in\widetilde{G},
\end{align*}
then the static analysis should add $z_1\xrightarrow{A}w_k$ to $\widetilde{G}$.

For example, while analyzing \texttt{test} in Figure~\ref{fig:example} with the library code on the right available, the analysis sees the path
\begin{align*}
\texttt{ob}\xrightarrow{\Store[\texttt{f}]}\texttt{this}_{\texttt{set}}\xrightarrow{\Alias}\texttt{this}_{\texttt{get}}\xrightarrow{\Load[\texttt{f}]}r_{\texttt{get}}.
\end{align*}
Since $\Transfer\xRightarrow{*}\Store[\texttt{f}]~\Transfer~\Load[\texttt{f}]$, the analysis adds edge $\texttt{ob}\xrightarrow{\Transfer}r_{\texttt{get}}$. As we describe below, the path specification $s_{\text{box}}$ shown in (\ref{eqn:expathspec}) ensures that this edge is added to $\widetilde{G}$ when the library code is unavailable.

\paragraph{{\bf\em Syntax.}}

Let $\mathcal{V}_{\text{prog}}$ be the variables in the program (i.e., excluding variables in the library), let $\mathcal{V}_m=\{p_m,r_m\}$ be the parameter and return value of library function $m\in\mathcal{M}$, and let $\mathcal{V}_{\text{path}}=\bigcup_{m\in\mathcal{M}}\mathcal{V}_m$ be the \emph{visible variables} (i.e., variables at the library interface). A \emph{path specification} is a sequence
\begin{align*}
z_1w_1z_2w_2...z_kw_k\in\mathcal{V}_{\text{path}}^*,
\end{align*}
where $z_i,w_i\in\mathcal{V}_{m_i}$ for library function $m_i\in\mathcal{M}$. We require that $w_i$ and $z_{i+1}$ are not both return values, and that $w_k$ is a return value. For clarity, we typically use the syntax
\begin{align}
\label{eqn:pathspec}
z_1\dashrightarrow w_1\to z_2\dashrightarrow...\dashrightarrow w_{k-1}\to z_k\dashrightarrow w_k.
\end{align}

\paragraph{{\bf\em Semantics.}}

\begin{figure*}
\scriptsize
\centering
\begin{tabular}{llll}
\hline
\\
\multicolumn{1}{c}{{\bf Candidate Code Fragment Specification}} & \multicolumn{1}{c}{{\bf Candidate Path Specification(s)}} & \multicolumn{2}{c}{{\bf Synthesized Unit Test}} \\\\
\hline
\\
\hspace{0.025in}
\begin{minipage}{2.0in}
\begin{verbatim}
class Box { // candidate specification
  Object f; // ghost field
  void set(Object ob) { f = ob; }
  Object get() { return f; } }
\end{verbatim}
\end{minipage}
&
\begin{minipage}{1.9in}
$\begin{array}{l}
\texttt{ob}\dashrightarrow\texttt{this}_{\texttt{set}}\to\texttt{this}_{\texttt{get}}\dashrightarrow r_{\texttt{get}}
\end{array}$
\end{minipage}
&
\begin{minipage}{1.7in}
\begin{verbatim}
boolean test() {
  Object in = new Object(); // o_in
  Box box = new Box(); // o_box
  box.set(in);
  Object out = box.get();
  return in == out; }
\end{verbatim}
\end{minipage}
&
\large\cmark
\hspace{0.025in}
\\\\
\hline
\\
\hspace{0.025in}
\begin{minipage}{2.0in}
\begin{verbatim}
class Box { // candidate specification
  Object f; // ghost field
  void set(Object ob) { f = ob; }
  Object clone() { return f; } }
\end{verbatim}
\end{minipage}
&
\begin{minipage}{1.9in}
$\begin{array}{l}
\texttt{ob}\dashrightarrow\texttt{this}_{\texttt{set}}\to\texttt{this}_{\texttt{clone}}\dashrightarrow r_{\texttt{clone}}
\end{array}$
\end{minipage}
&
\begin{minipage}{1.7in}
\begin{verbatim}
boolean test() {
  Object in = new Object(); // o_in
  Box box = new Box(); // o_box
  box.set(in);
  Object out = box.clone();
  return in == out; }
\end{verbatim}
\end{minipage}
&
\large\xmark
\hspace{0.025in}
\\\\
\hline
\\
\hspace{0.025in}
\begin{minipage}{2.0in}
\begin{verbatim}
class Box { // candidate specification
  Object f; // ghost field
  void set(Object ob) { f = ob; }
  Object get() { return f; }
  Box clone() {
    Box b = new Box(); // ~o_clone
    b.f = f;
    return b; } }
\end{verbatim}
\end{minipage}
&
\begin{minipage}{1.9in}
$\begin{array}{ll}
\texttt{ob}\dashrightarrow\texttt{this}_{\texttt{set}}&\hspace{-0.1in}\blue{(}\to\texttt{this}_{\texttt{clone}}\dashrightarrow r_{\texttt{clone}}\blue{)^*} \\
&\hspace{-0.1in}\to\texttt{this}_{\texttt{get}}\dashrightarrow r_{\texttt{get}}
\end{array}$
\end{minipage}
&
\begin{minipage}{1.7in}
\begin{verbatim}
boolean test0() {
  Object in = new Object(); // o_in
  Box box0 = new Box(); // o_box
  box0.set(in);
  Object out = box0.get();
  return in == out; }
\end{verbatim}
\begin{verbatim}
boolean test1() {
  Object in = new Object(); // o_in
  Box box0 = new Box(); // o_box
  box0.set(in);
  Box box1 = box0.clone();
  Object out = box1.get();
  return in == out; }
\end{verbatim}
\begin{verbatim}
boolean test2() {
  Object in = new Object(); // o_in
  Box box0 = new Box(); // o_box
  box0.set(in);
  Box box1 = box0.clone();
  Box box2 = box1.clone();
  Object out = box2.get();
  return in == out; }
\end{verbatim}
\begin{verbatim}
...
\end{verbatim}
\end{minipage}
&
\large\cmark
\hspace{0.025in}
\\\\
\hline
\end{tabular}
\caption{
Examples of hypothesized library implementations (left column), an equivalent set of path specifications (middle column), and the synthesized unit tests to check the precision of these specifications (right column), with a check mark $\cmark$ (indicating that the unit tests pass) or a cross mark $\xmark$ (indicating that the unit tests fail).
}
\label{fig:examples}
\end{figure*}

Given path specification (\ref{eqn:pathspec}), for each $i\in[k]$, define the nonterminal $A_i$ in the grammar $C_{\text{pt}}$ to be
\begin{align*}
A_i=
\begin{cases}
\Transfer&\text{if }w_i=r_{m_i}\text{ and }z_{i+1}=p_{m_{i+1}} \\
\Alias&\text{if }w_i=p_{m_i}\text{ and }z_{i+1}=p_{m_{i+1}} \\
\overline{\Transfer}&\text{if }w_i=p_{m_i}\text{ and }z_{i+1}=r_{m_{i+1}}.
\end{cases}
\end{align*}
Also, define the nonterminal $A$ by
\begin{align*}
A=\begin{cases}
\Transfer&\text{ if }z_1=p_{m_1} \\
\Alias&\text{ if }z_1=r_{m_1}.
\end{cases}
\end{align*}
Then, the path specification corresponds to adding a rule
\begin{align*}
\left(\bigwedge_{i=1}^{k-1}w_i\xrightarrow{A_i}z_{i+1}\in\widetilde{G}\right)\Rightarrow(z_1\xrightarrow{A}w_k\in\widetilde{G})
\end{align*}
to the static points-to analysis. The rule also adds the backwards edge $w_k\xrightarrow{\overline{A}}z_1$ to $\widetilde{G}$, but we omit it for clarity. We refer to the antecedent of this rule as the \emph{premise} of the path specification, and the consequent of this rule as the \emph{conclusion} of the path specification. Continuing our example, the path specification $s_{\text{box}}$ shown in (\ref{eqn:expathspec}) has semantics
\begin{align*}
(\texttt{this}_{\texttt{set}}\xrightarrow{\Alias}\texttt{this}_{\texttt{get}})\Rightarrow(\texttt{ob}\xrightarrow{\Transfer}r_{\texttt{get}}).
\end{align*}
This rule says that if the static analysis computes the edge $\texttt{this}_{\texttt{set}}\xrightarrow{\Alias}\texttt{this}_{\texttt{get}}\in\widetilde{G}$, then it must add $\texttt{ob}\xrightarrow{\Transfer}r_{\texttt{get}}$ to $\widetilde{G}$. For example, this rule is applied in Figure~\ref{fig:graph} (right) to compute $\texttt{ob}\xrightarrow{\Transfer}r_{\texttt{get}}$.

The middle column of Figure~\ref{fig:examples} shows examples of path specifications, and the first column shows equivalent code fragment specifications (the last column is described below). The specifications on the first and last rows are precise, whereas the specification on the second row is imprecise.

\paragraph{{\bf\em Soundness and precision.}}

Let $\widetilde{G}_*(P)$ denote the true set of relations for a program $P$ (i.e., relations that hold dynamically for some execution of $P$); note that because the library code is omitted from analysis, we only include relations between program variables in $\widetilde{G}_*(P)$. Then, given path specifications $S$, let $\widetilde{G}(P,S)$ denote the points-to edges computed using $S$ for $P$, let $\widetilde{G}_+(P,S)=\widetilde{G}(P,S)\setminus\widetilde{G}_*(P)$ be the false positives, and let $\widetilde{G}_-(P,S)=\widetilde{G}_*(P)\setminus\widetilde{G}(P,S)$ be the false negatives.

Specification set $S$ is \emph{sound} if $\widetilde{G}_-(P,S)=\varnothing$. We say $S$ and $S'$ are \emph{equivalent} (written $S\equiv S'$) if for every program $P$, $\widetilde{G}(P,S)=\widetilde{G}(P,S')$. Finally, $S$ is \emph{precise} if for every sound $S'$, $S'\cup S\equiv S'$. In other words, using $S$ computes no false positive points-to edges compared to any sound set $S'$. We say a single path specification $s$ is precise if $\{s\}$ is precise. We have the following result, which follows by induction:
\begin{theorem}
\label{thm:logic}
\rm
For any set $S$ of path specifications, if each $s\in S$ is precise, then $S$ is precise.
\end{theorem}

\paragraph{{\bf\em Witnesses.}}

Our algorithm needs to synthesize unit tests that check whether a candidate path specification $s$ is precise. In particular, a program $P$ is a \emph{potential witness} for a candidate path specification $s$ if $P$ returns true only if $s$ is precise. If $P$ is a potential witness for $s$, and upon execution, $P$ in fact returns true, then we say $P$ is a \emph{witness} for $s$. The last column of Figure~\ref{fig:examples} shows potential witnesses for the candidate path specification in the middle column; a green check indicates the potential witness returns true (we say $P$ \emph{passes}), and a red check indicates it returns false or raises an exception (we say $P$ \emph{fails}). In particular, the synthesized unit test correctly rejects the imprecise specification on the second row.

Note that $P$ may return false even if $s$ is precise. This property is inevitable since executions are underapproximations; we show empirically that if $s$ is precise, then typically the potential witness synthesized by our algorithm passes.

\paragraph{{\bf\em Soundly and precisely modeling library code.}}

It is not obvious that path specifications are sufficiently expressive to precisely model library code. In this section, we show that path specifications are in fact sufficiently expressive to do so in the case of Andersen's analysis (and its cloning-based context- and object-sensitive extensions). More precisely, for any implementation of the library, there exists a (possibly infinite) set of path specifications such that the points-to sets computed using path specifications are both sound and at least as precise as analyzing the library implementation:
\begin{theorem}
\label{THM:EQUIV}
\rm
Let $\widetilde{G}(P)$ be the set of points-to edges computed for program $P$ assuming the library code is available. Then, there exists a set $S$ of path specifications such that for every program $P$, $\widetilde{G}(P,S)$ is sound and $\widetilde{G}(P,S)\subseteq\widetilde{G}(P)$.
\end{theorem}
We give a proof in Appendix~\ref{sec:mainproof}. Note that the set $S$ of path specifications may be infinite. This infinite blowup is unavoidable since we want the ability to test the precision of an individual path specification. In particular, the library implementation may exhibit effects that require infinitely many unit tests to check precision (e.g., the path specifications shown on the third row of Figure~\ref{fig:examples}).

\paragraph{{\bf\em Regular sets of path specifications.}}

Since the library implementation may correspond to an infinite set of path specifications, we need a mechanism for describing such sets. In particular, since a path specification is a sequence $s\in\mathcal{V}_{\text{path}}^*$, we can think of a set $S$ of path specifications as a formal language $S\subseteq\mathcal{V}_{\text{path}}^*$ over the alphabet $\mathcal{V}_{\text{path}}$. Then, we can express an infinite set of path specifications using standard representations such as regular expressions or context-free grammars.

We make the empirical observation that the library implementation is equivalent to a regular set $S_*$ of path specifications (i.e., $S_*$ is sound and precise). There is no particular reason that this fact should be true, but it holds for all the Java library functions we have examined so far. For example, consider the code fragment shown in the first column of the third row of Figure~\ref{fig:examples}. This specification corresonds to the set of path specifications shown as a regular expression in the middle column of the same line (tokens in the regular expression are highlighted in blue for clarity).

\paragraph{{\bf\em Static points-to analysis with path specifications.}}

To perform static points-to analysis with regular sets of path specifications, we convert the path specifications into equivalent code fragments and then analyze the client along with the code fragments; see Appendix~\ref{sec:pathanalysis} for details.

\section{Specification Inference Algorithm}
\label{sec:inference}

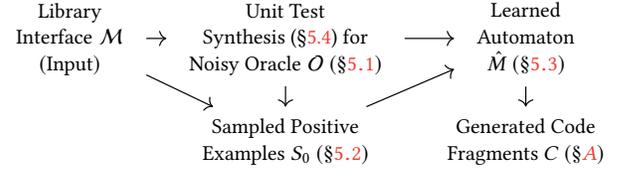
\begin{figure}
\footnotesize
\[
\begin{tikzcd}[column sep=1.0em,row sep=1.0em]
\begin{array}{c}\text{Library}\\\text{Interface}~\mathcal{M}\\(\text{Input})\end{array} \arrow[r] \arrow[rd]
& \begin{array}{c}\text{Unit Test}\\\text{Synthesis (\S\ref{sec:testsynthesis}) for}\\\text{Noisy Oracle}~\mathcal{O}~(\S\ref{sec:formulation})\end{array} \arrow[d] \arrow[r]
& \begin{array}{c}\text{Learned}\\\text{Automaton}\\\hat{M}~(\S\ref{sec:rpni})\end{array} \arrow[d]
\\
& \begin{array}{c}\text{Sampled Positive}\\\text{Examples}~S_0~(\S\ref{sec:sample})\end{array} \arrow[ru]
& \begin{array}{c}\text{Generated Code}\\\text{Fragments}~C~(\S\ref{sec:pathanalysis})\end{array}
\end{tikzcd}
\]
\caption{An overview of our specification inference system. The section describing each component is in parentheses.}
\label{fig:system}
\end{figure}

We describe our algorithm for inferring path specifications. Our system is summarized in Figure~\ref{fig:system}, which also shows the section where each component is described in detail.

\subsection{Overview}
\label{sec:formulation}

Let the \emph{target language} $S_*\subseteq\mathcal{V}_{\text{path}}^*$ be the set of all path specifications that are precise. By Theorem~\ref{thm:logic}, $S_*$ is precise. The goal of our algorithm is to infer a set of path specifications that approximates $S_*$ as closely as possible.

\paragraph{{\bf\em Inputs.}}

Recall that our algorithm is given two inputs: (i) the library interface, and (ii) blackbox access to the library functions. We use these two inputs to construct the noisy oracle and positive examples as describe below.

\paragraph{{\bf\em Noisy oracle.}}

Given a path specification $s$, the \emph{noisy oracle} $\mathcal{O}:\mathcal{V}_{\text{path}}^*\to\{0,1\}$ (i) always returns $0$ if $s$ is imprecise, and (ii) ideally returns $1$ if $s$ is precise (but may return $0$).
\footnote{While our implementation of the oracle is deterministic, our specification inference algorithm can also make use of a stochastic oracle (as long as it satisfies these two properties).}
This oracle is implemented by synthesizing a potential witness $P$ for $s$ and returning the result of executing $P$. We describe how we synthesize a witness for $s$ in Section~\ref{sec:testsynthesis}.

\paragraph{{\bf\em Positive examples.}}

Phase one of our algorithm constructs a set of positive examples: our algorithm randomly samples candidate path specifications $s\sim\mathcal{V}_{\text{path}}^*$, and then uses $\mathcal{O}$ to determine whether each $s$ is precise. More precisely, given a set $S=\{s\sim\mathcal{V}_{\text{path}}^*\}$ of random samples, it constructs positive examples $S_0=\{s\in S\mid\mathcal{O}(s)=1\}$. We describe how we sample $s\sim\mathcal{V}_{\text{path}}^*$ in Section~\ref{sec:sample}.

\paragraph{{\bf\em Language inference problem.}}

Phase two of our algorithm inductively generalizes $S_0$ to a regular set of path specifications. We formulate this problem as follows:
\begin{definition}
\rm
The \emph{language inference problem} is to, given the noisy oracle $\mathcal{O}$ and the positive examples $S_0\subseteq S_*$, infer a language $\hat{S}$ that approximates $S_*$ as closely as possible.
\end{definition}
The approximation quality of $\hat{S}$ compared to $S_*$ must take into account both the false positive rate and the false negative rate. Intuitively, we prioritize minimizing false positives over minimizing false negatives---i.e., we aim to maximize the size of $\hat{S}$ subject to $\hat{S}\subseteq S_*$; however, $\hat{S}$ and $S_*$ may be infinitely large. In our evaluation, we use a heuristic to measure approximation quality; see Section~\ref{sec:eval} for details.

In Section~\ref{sec:rpni}, we describe our algorithm for solving the language inference problem. It outputs a regular language $\hat{S}=\mathcal{L}(\hat{M})$, where $\hat{M}$ is a finite state automaton---e.g., given
\begin{align*}
S_0=\{\texttt{ob}~\texttt{this}_{\texttt{set}}~\texttt{this}_{\texttt{clone}}~r_{\texttt{clone}}~\texttt{this}_{\texttt{get}}~r_{\texttt{get}}\},
\end{align*}
our language inference algorithm returns an automaton encoding the regular language
\begin{align*}
\texttt{ob}~\texttt{this}_{\texttt{set}}~\blue{(}\texttt{this}_{\texttt{clone}}~r_{\texttt{clone}}\blue{)^*}~\texttt{this}_{\texttt{get}}~r_{\texttt{get}}.
\end{align*}

\subsection{Sampling Positive Examples}
\label{sec:sample}

We sample a path specification $s\in\mathcal{V}_{\text{path}}^*$ by building it one variable at a time, starting from $s=\varepsilon$ (where $\varepsilon$ denotes the empty string). At each step, we ensure that $s$ satisfies the path specification constraints, i.e., (i) $z_i$ and $w_i$ are parameters or return values of the same library function, (ii) $w_i$ and $z_{i+1}$ are not both return values, and (iii) the last variable $w_k$ is a return value. In particular, given current sequence $s$, the set $\mathcal{T}(s)\subseteq\mathcal{V}_{\text{path}}\cup\{\phi\}$ of choices for the next variable (where $\phi$ indicates to terminate and return $s$) is:
\begin{itemize}
\item If $s=z_1w_1z_2...z_i$, then the choices for $w_i$ are $\mathcal{T}(s)=\{p_m,r_m\}$, where $z_i\in\{p_m,r_m\}$.
\item If $s=z_1w_1z_2...z_iw_i$, and $w_i$ is a parameter, then the choices for $z_{i+1}$ are $\mathcal{T}(s)=\mathcal{V}_{\text{path}}$.
\item If $s=z_1w_1z_2...z_iw_i$, and $w_i$ is a return value, then the choices for $z_{i+1}$ are
\begin{align*}
\mathcal{T}(s)=\{z\in\mathcal{V}_{\text{path}}\mid z\text{ is a parameter}\}\cup\{\phi\}.
\end{align*}
\end{itemize}
At each step, our algorithm samples $x\sim\mathcal{T}(s)$, and either constructs $s'=sx$ and continues if $x\neq\phi$ or returns $s$ if $x=\phi$. We consider two sampling strategies.

\paragraph{{\bf\em Random sampling.}}

We  choose $x\sim\mathcal{T}(s)$ uniformly at random at every step.

\paragraph{{\bf\em Monte Carlo tree search.}}

We can exploit the fact that certain choices $x\in\mathcal{T}(s)$ are much more likely to yield a precise path specification than others. To do so, note that our search space is structured as a tree. Each edge in this tree is labeled with a symbol $x\in\mathcal{V}_{\text{path}}\cup\{\phi\}$; then, we can associate each node $N$ in the tree with the sequence $s_N\in(\mathcal{V}_{\text{path}}\cup\{\phi\})^*$ obtained by traversing the tree from the root to $N$ and collecting the labels along the edges (if $N$ is the root node, then $s_N=\varepsilon$). Given an internal node $N$ with corresponding sequence $s_N$, its children are determined by $\mathcal{T}$ as follows:
\begin{align*}
  \{N\xrightarrow{x}N'\mid x\in\mathcal{T}(s_N)\}.
\end{align*}
Therefore, a leaf node $L$ corresponds to a sequence of the form $s_L=x_1...x_k\phi$, which in turn corresponds to a candidate path specification $s=x_1...x_n$.
Thus, we can sample $x\sim\mathcal{T}(s)$ using Monte Carlo tree search (MCTS)~\cite{kocsis2006bandit}, a search algorithm that learns over time which choices are more likely to succeed. In particular, MCTS keeps track of a score $Q(N,x)$ for every visited node $N$ and every $x\in\mathcal{T}(s_N)$. Then, the choices are sampled according to the distribution
\begin{align*}
\text{Pr}[x\mid N]&=\frac{1}{Z}e^{Q(N,x)}\sss\text{where}\sss Z=\sum_{x'\in\mathcal{T}(s_N)}e^{Q(N,x')}.
\end{align*}
Whenever a candidate $s=x_1...x_k$ is found, we increase the score $Q(x_1...x_i,x_{i+1})$ (for each $0\le i<k$) if $s$ is a positive example ($\mathcal{O}(s)=1$) and decrease it otherwise ($\mathcal{O}(s)=0$):
\begin{align*}
Q(x_1...x_i,x_{i+1})\gets(1-\alpha)Q(x_1...x_i,x_{i+1})+\alpha\mathcal{O}(s).
\end{align*}
We choose the \emph{learning rate} $\alpha$ to be $\alpha=1/2$.

\subsection{Language Inference Algorithm}
\label{sec:rpni}

Our language inference algorithm is based on RPNI~\cite{oncina1992identifying}.
\footnote{We also considered the $L^*$ algorithm~\cite{angluin1987learning}; however, the $L^*$ algorithm depends on an \emph{equivalence oracle} that reports whether a candidate language is correct, which is unavailable in our setting. It is possible to approximate the equivalence oracle using sampling, but in our experience, this approximation is very poor and can introduce substantial imprecision.}
In particular, we modify RPNI to leverage access to the noisy oracle---whereas RPNI takes as input a set of negative examples, we use the oracle to generate them on-the-fly. Our algorithm learns a regular language $\hat{S}=\mathcal{L}(\hat{M})$ represented by the (nondeterministic) finite state automaton (FSA) $\hat{M}=(Q,\mathcal{V}_{\text{path}},\delta,q_{\text{init}},Q_{\text{fin}})$, where $Q$ is the set of states, $\delta:Q\times\mathcal{V}_{\text{path}}\to2^Q$ is the transition function, $q_{\text{init}}\in Q$ is the start state, and $Q_{\text{fin}}\subseteq Q$ are the accept states. If there is a single accept state, we denote it by $q_{\text{fin}}$. We denote transitions $q\in\delta(p,\sigma)$ by $p\xrightarrow{\sigma}q$.

Our algorithm initializes $\hat{M}$ to be the FSA representing the finite language $S_0$. In particular, it initializes $\hat{M}$ to be the prefix tree acceptor~\cite{oncina1992identifying}, which is the FSA where the underlying transition graph is the prefix tree of $S_0$, the start state is the root of this prefix tree, and the accept states are the leaves of this prefix tree.

Then, our algorithm iteratively considers \emph{merging} pairs of states of $\hat{M}$. More precisely, given two states $p,q\in Q$ (without loss of generality, assume $q\not=q_{\text{init}}$), $\textsf{Merge}(\hat{M},q,p)$ is the FSA obtained by (i) replacing transitions
\begin{align*}
(r\xrightarrow{\sigma}q)\text{ becomes }(r\xrightarrow{\sigma}p),\hspace{0.2in}
(q\xrightarrow{\sigma}r)\text{ becomes }(p\xrightarrow{\sigma}r),
\end{align*}
(ii) adding $p$ to $Q_{\text{fin}}$ if $q\in Q_{\text{fin}}$, and (iii) removing $q$ from $Q$.

Our algorithm iterates once over all the states $Q$; we describe how a single iteration proceeds. Let $q$ be the state being processed in the current step, let $Q_0$ be the states that have been processed so far but not removed from $Q$, and let $\hat{M}$ be the current FSA. For each $p\in Q_0$, our algorithm checks whether merging $q$ and $p$ adds imprecise path specifications to the language $\mathcal{L}(\hat{M})$; if not, it greedily performs the merge. More precisely, for each $p\in Q_0$, our algorithm constructs
\begin{align*}
M_{\text{diff}}=\textsf{Merge}(\hat{M},q,p)\setminus\hat{M},
\end{align*}
which represents the set of strings that are added to $\mathcal{L}(\hat{M})$ if $q$ and $p$ are merged. Then, for each $s\in M_{\text{diff}}$ up to some maximum length $N$ (we take $N=8$), our algorithm queries $\mathcal{O}(s)$. If all queries pass (i.e., $\mathcal{O}(s)=1$), then our algorithm greedily accepts the merge, i.e., $\hat{M}\gets\textsf{Merge}(\hat{M},q,p)$ and continues to the next $q\in Q$. Otherwise, it considers merging $q$ with the next $p\in Q_0$. Finally, if $q$ is not merged with any state $p\in Q_0$, then our algorithm does not modify $\hat{M}$ and adds $q$ to $Q_0$. Once it has completed a pass over all states in $Q$, our algorithm returns $\hat{M}$. For example, suppose our language learning algorithm is given a single positive example
\begin{align*}
\texttt{ob}~\texttt{this}_{\texttt{set}}~\texttt{this}_{\texttt{clone}}~r_{\texttt{clone}}~\texttt{this}_{\texttt{get}}~r_{\texttt{get}}.
\end{align*}
Then, our algorithm constructs the finite state automaton
\[
\begin{footnotesize}
q_{\text{init}}\xrightarrow{\texttt{ob}}q_1\xrightarrow{\texttt{this}_{\texttt{set}}}q_2
\xrightarrow{\texttt{this}_{\texttt{clone}}}q_3\xrightarrow{r_{\texttt{clone}}}q_4
\xrightarrow{\texttt{this}_{\texttt{get}}}q_5\xrightarrow{r_{\texttt{get}}}q_{\text{fin}}.
\end{footnotesize}
\]
Our algorithm fails to merge $q_{\text{init}}$, $q_1$, $q_2$, or $q_3$ with any previous states. It then tries to merge $q_4$ with each state $\{q_{\text{init}},q_1,q_2,q_3\}$; the first two merges fail, but merging $q_4$ with $q_2$ produces
\[
\begin{tikzcd}[column sep=3.5em,row sep=2em]
q_{\text{init}} \arrow[r,"{\texttt{ob}}"] 
& q_1 \arrow[r,"{\texttt{this}_{\texttt{set}}}"]
& q_2 \arrow[r,"{\texttt{this}_{\texttt{get}}}"] \arrow[bend right]{d}[xshift=-8ex, yshift=-2ex]{\texttt{this}_{\texttt{clone}}}
& q_4 \arrow[r,"{r_{\texttt{get}}}"]
& q_{\text{fin}}. \\
&& q_3 \arrow[bend right]{u}[xshift = 5ex, yshift=-2ex]{r_{\texttt{clone}}}
\end{tikzcd}
\]
Then, the specifications of length at most $N$ in $M_{\text{diff}}$ are
\begin{align*}
&\texttt{ob}~\texttt{this}_{\texttt{set}}~\blue{(}\texttt{this}_{\texttt{clone}}~r_{\texttt{clone}}\blue{)^0}~\texttt{this}_{\texttt{get}}~r_{\texttt{get}} \\
&\texttt{ob}~\texttt{this}_{\texttt{set}}~\blue{(}\texttt{this}_{\texttt{clone}}~r_{\texttt{clone}}\blue{)^1}~\texttt{this}_{\texttt{get}}~r_{\texttt{get}} \\
& ... \\
&\texttt{ob}~\texttt{this}_{\texttt{set}}~\blue{(}\texttt{this}_{\texttt{clone}}~r_{\texttt{clone}}\blue{)^N}~\texttt{this}_{\texttt{get}}~r_{\texttt{get}},
\end{align*}
all of which are accepted by $\mathcal{O}$. Therefore, our algorithm greedily accepts this merge and continues. The remaining merges fail and our algorithm returns this automaton.

\subsection{Unit Test Synthesis}
\label{sec:testsynthesis}

\begin{figure}
\small
\centering
\begin{tabular}{ll}
\hline \vspace{-0.1in} \\
$\begin{array}{l}\text{input}\end{array}$ &
$\begin{array}{rl}
\texttt{ob}\dashrightarrow\texttt{this}_{\texttt{set}}\hspace{-0.1in}&\to\texttt{this}_{\texttt{clone}}\to r_{\texttt{clone}} \\
&\to\texttt{this}_{\texttt{get}}\to r_{\texttt{get}}
\end{array}$
\vspace{0.03in} \\
\hline \vspace{-0.1in} \\
$\begin{array}{l}\text{skeleton}\end{array}$ &
$\begin{array}{l}
\textttb{??.set(??);}\\
\textttb{?? = ??.clone();}\\
\textttb{?? = ??.get();}\\
\end{array}$
\vspace{0.03in} \\
\hline \vspace{-0.1in} \\
$\begin{array}{l}\text{fill holes}\end{array}$ &
$\begin{array}{l}
\textttb{box}\texttt{.set(}\textttb{in}\texttt{);}\\
\textttb{Box boxClone}\texttt{ = }\textttb{box}\texttt{.clone();}\\
\textttb{Object out}\texttt{ = }\textttb{boxClone}\texttt{.get();}
\end{array}$
\vspace{0.03in} \\
\hline \vspace{-0.1in} \\
$\begin{array}{l}\text{initialization}\\\text{\& scheduling}\end{array}$ &
$\begin{array}{l}
\textttb{Object in = new Object();}\\
\textttb{Box box = new Box()}\\
\texttt{box.set(in);}\\
\texttt{Box boxClone = box.clone();}\\
\texttt{Object out = boxClone.get();}\\
\textttb{return in == out};\\
\end{array}$
\vspace{0.03in} \\
\hline
\end{tabular}
\caption{Steps in the witness synthesis algorithm for a candidate path specification for \texttt{Box}. Code added at each step is highlighted in blue. Scheduling is shown in the same line as initialization.}
\label{fig:testsynthesis}
\end{figure}

We describe how we synthesize a unit test that is a potential witness for a given specification
\begin{align*}
s=(z_1\dashrightarrow w_1\to...\to z_k\dashrightarrow w_k),
\end{align*}
relegating details to Appendix~\ref{sec:testsynthesisappendix}. Figure~\ref{fig:testsynthesis} shows how the unit test synthesis algorithm synthesizes a unit test for the candidate specification $s_{\text{box}}$.

Recall that the semantics of $s$ are $\left(\bigwedge_{i=1}^{k-1}w_i\xrightarrow{A_i}z_{i+1}\in\widetilde{G}\right)\Rightarrow(z_1\xrightarrow{A}w_k\in\widetilde{G})$. Then, consider a program $P$ that satisfies the following properties:
\begin{itemize}
\item The conclusion of $s$ does not hold statically for $P$ with empty specifications, i.e., $z_1\xrightarrow{A}w_k\not\in\widetilde{G}(P,\varnothing)$.
\item The premise of $s$ holds for $P$, i.e., $w_i\xrightarrow{A_i}z_{i+1}\in\widetilde{G}(P,\{s\})$ for each $i\in[k-1]$.
\item For every set $S$ of path specifications, if $z_1\xrightarrow{A}w_k\in\widetilde{G}(P,S)$, then $S\cup\{s\}$ is equivalent to $S$.
\end{itemize}
Intuitively, if program $P$ is a potential witness for path specification $s$ with premise $\psi$ and conclusion $\phi=(e\in\widetilde{G})$, then $s$ is the only path specification that can be used by the static analysis to compute relation $e$ for $P$. Thus, if $P$ witnesses $s$, then $s$ is guaranteed to be precise. In particular, we have the following important guarantee for the unit test synthesis algorithm (see Appendix~\ref{sec:witnessproof} for a proof):
\begin{theorem}
\label{THM:TESTSYNTHESIS}
\rm
The unit test $P$ synthesized for path specification $s$ is a potential witness for $s$.
\end{theorem}

\paragraph{{\bf\em Skeleton construction.}}

Our algorithm first constructs the \emph{skeleton} of the unit test. In particular, a witness for $s$ must include a call to each function $m_1,...,m_k$, where the variables $z_i,w_i\in\mathcal{V}_{m_i}$ are parameters or return values of $m_i$, since the graph $G$ extracted from the unit test must by definition contain edges connecting the $w_i$ to $z_{i+1}$. For each function call $y\gets m(x)$, the argument $x$ and the left-hand side variable $y$ are left as holes \texttt{??} to be filled in subsequent steps.

\paragraph{{\bf\em Fill holes.}}

Second, our algorithm fills the holes in the skeleton corresponding to reference variables. In particular, for each pair of function calls $\texttt{??}_{y,i}\gets m_i(\texttt{??}_{x,i})$ and $\texttt{??}_{y,i+1}\gets m_{i+1}(\texttt{??}_{x,i+1})$, it fills the holes $\texttt{??}_{y,i}$ and $\texttt{??}_{x,i+1}$ depending on the edge $w_i\xrightarrow{A_i}z_{i+1}$:
\begin{itemize}
\item {\bf Case $A_i=\Transfer$:} In this case, $w_i$ is a return value and $z_{i+1}$ is a parameter. Thus, the algorithm fills $\texttt{??}_{y,i}$ and $\texttt{??}_{x,i+1}$ with the same fresh variable $x$.
\item {\bf Case $A_i=\overline{\Transfer}$:} This case is analogous to the case $A_i=\Transfer$.
\item {\bf Case $A_i=\Alias$:} In this case, $w_i$ and $z_{i+1}$ are both parameters. Thus, the algorithm fills $\texttt{??}_{y,i}$ and $\texttt{??}_{x,i+1}$ with the same fresh variable $x$, and additionally adds to the test an allocation statement $x\gets X()$.
\end{itemize}
As shown in the proof of Theorem~\ref{THM:TESTSYNTHESIS}, the added statements ensure that $P$ is a potential witness for $s$.

\paragraph{{\bf\em Initialization.}}

Third, our algorithm initializes the remaining reference and primitive variables in the unit test. In particular, function calls $y\gets m_i(x)$ may have additional parameters that need to be filled.
For Theorem~\ref{THM:TESTSYNTHESIS} to hold, the remaining reference variables must be initialized to \texttt{null}.

However, this approach is likely to synthesize unit tests that fail (by raising exceptions) even when $s$ is precise. Therefore, we alternatively use a heuristic where we allocate a fresh variable for each reference variable. For allocating reference variables that are passed as arguments to constructors, we have to be careful to avoid infinite recursion; for example, a constructor \texttt{Integer(Integer i)} should be avoided. Our algorithm uses a shortest-path algorithm to generate the smallest possible initialization statements; see Appendix~\ref{sec:initializationappendix} for details. With this approach, we can no longer guarantee that $P$ is a witness and our oracle may be susceptible to false positives. In our evaluation, we show that  this heuristic substantially improves recall with no reduction in precision.

Primitive variables can be initialized arbitrarily for Theorem~\ref{THM:TESTSYNTHESIS} to hold, but this choice affects whether $P$ is a witness when $s$ is precise. We initialize primitive variables using default values (\texttt{0} for numeric variables and \texttt{true} for boolean variables) that work well in practice.

\paragraph{{\bf\em Scheduling.}}

Fourth, our algorithm determines the ordering of the statements in the unit test. There are many possible choices of statement ordering, which affect whether $P$ is a witness when $s$ is precise. There are \emph{hard constraints} on the ordering (in particular, a variable must be defined before it is used) and \emph{soft constraints} (in particular, statements corresponding to edges $w_i\to z_{i+1}$ for smaller $i$ should occur earlier in $P$). Our scheduling algorithm produces an ordering that satisfies the hard constraints while trying to satisfy as many soft constraints as possible. It uses a greedy strategy, i.e., it orders the statements sequentially from first to last, choosing at each step the statement that satisfies all the hard constraints and the most soft constraints; see Appendix~\ref{sec:scheduling}.

\begin{figure}
\includegraphics[width=0.4\textwidth]{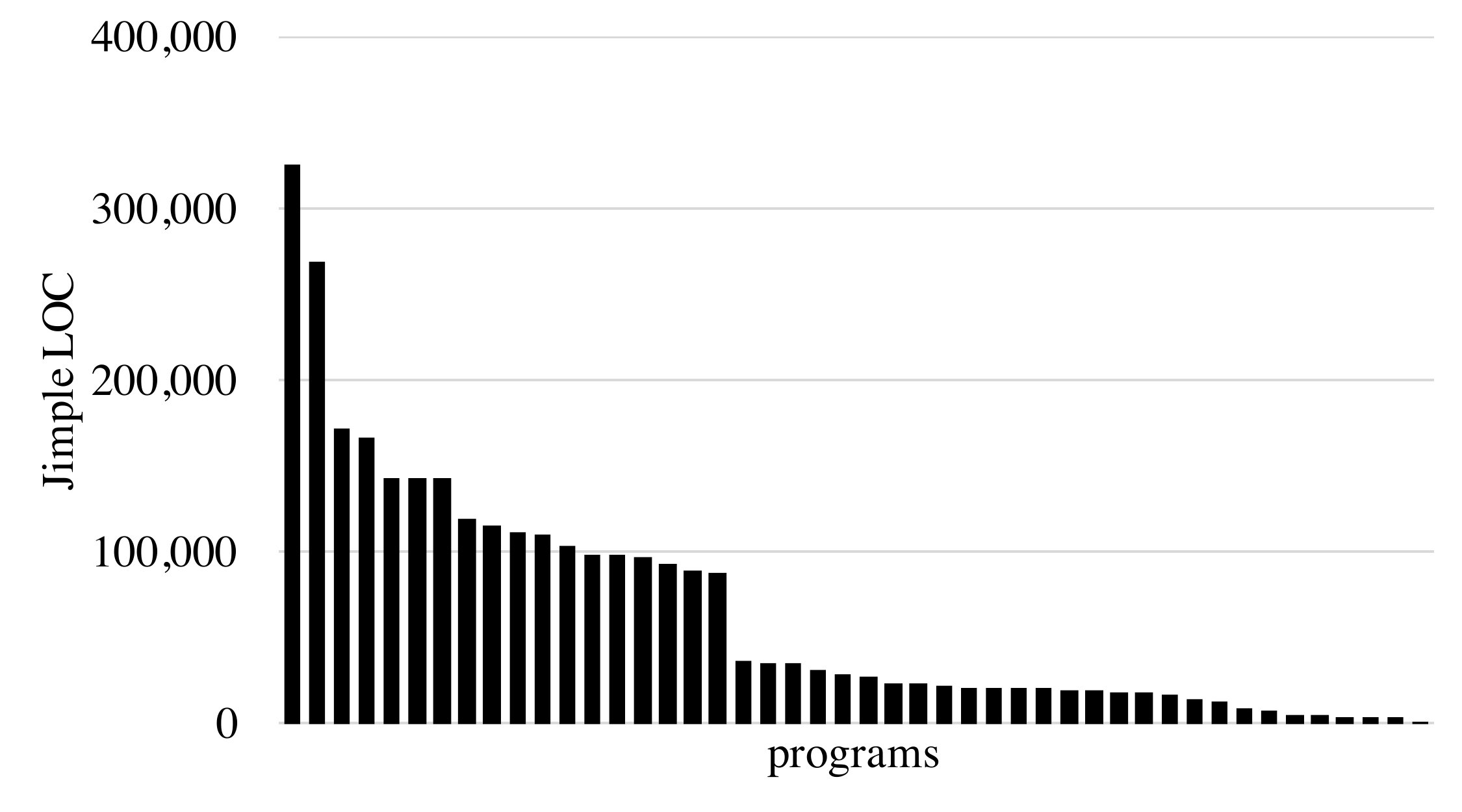}
\caption{Jimple lines of code of the apps in our benchmark.}
\label{fig:sizes}
\end{figure}


\begin{figure*}
\begin{tabular}{ccc}
\includegraphics[width=0.3\textwidth]{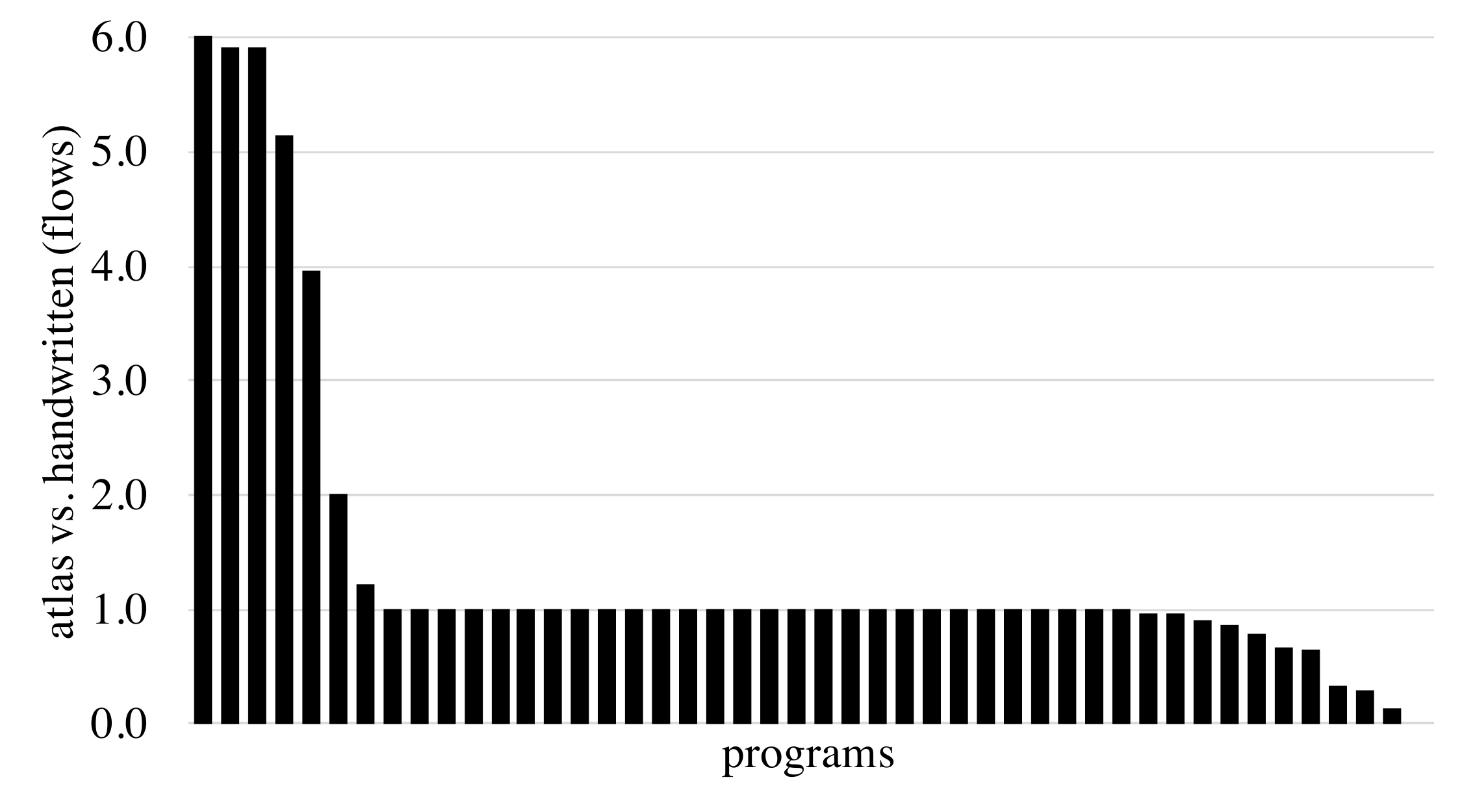}
&\includegraphics[width=0.3\textwidth]{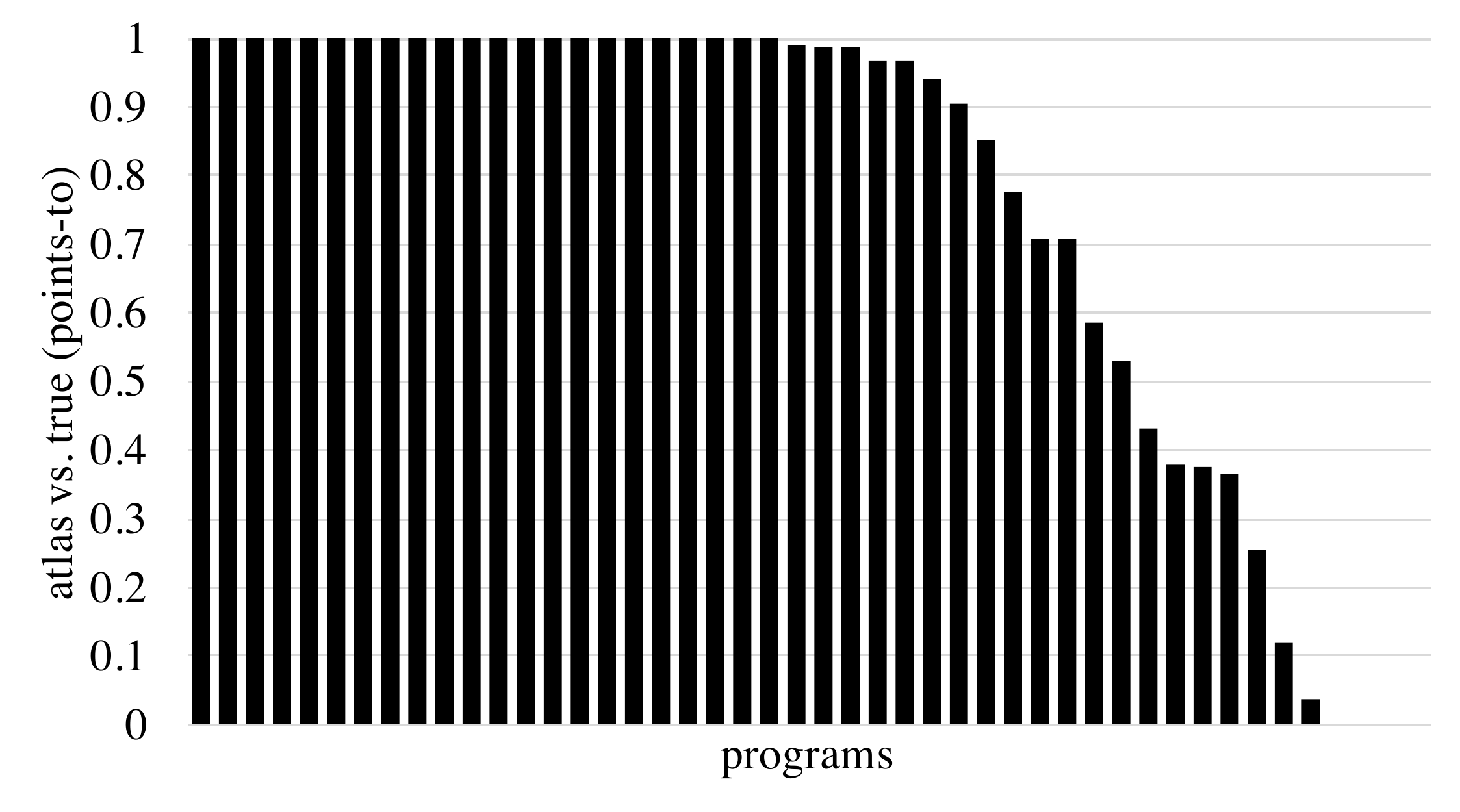}
&\includegraphics[width=0.3\textwidth]{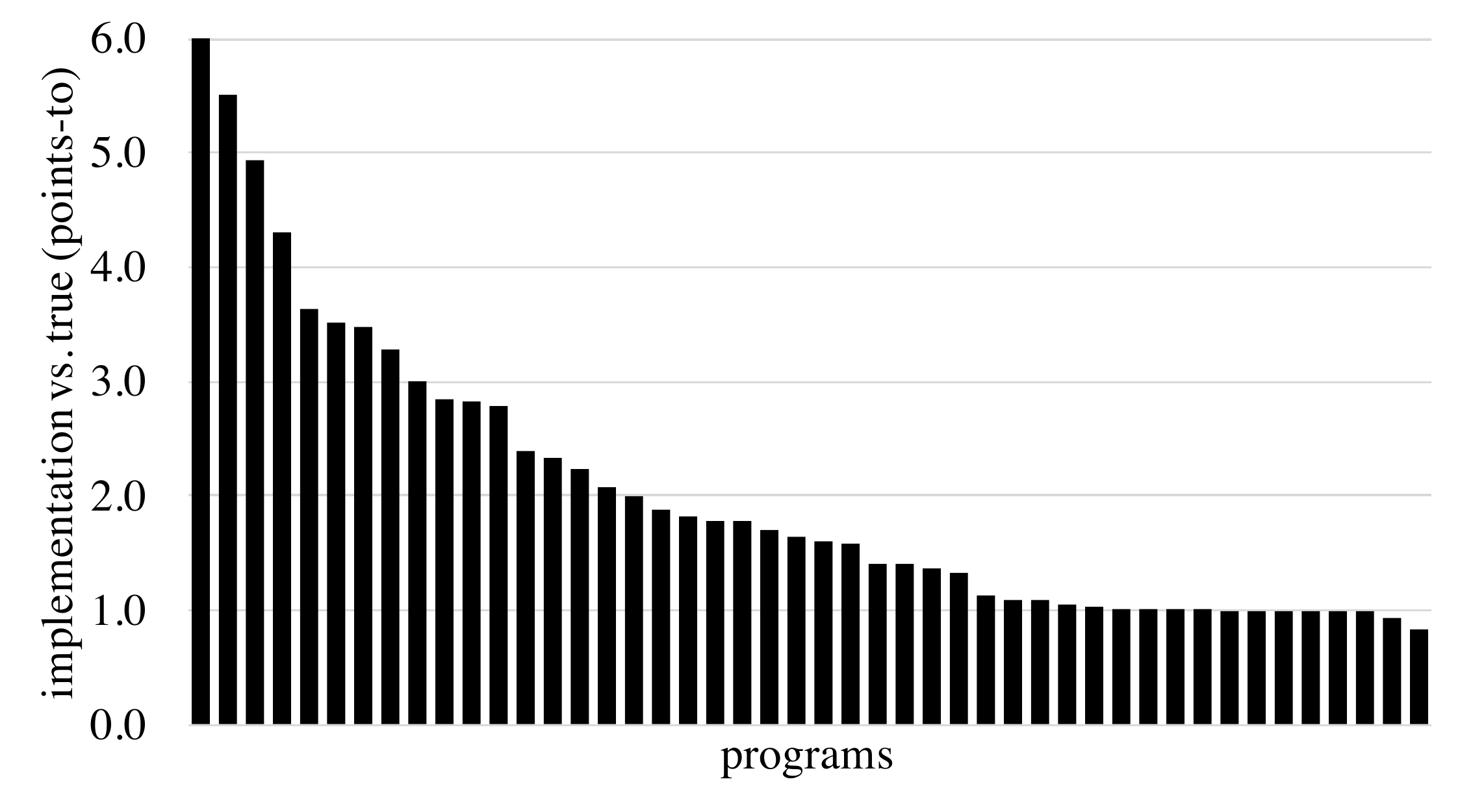} \\
(a) & (b) & (c)
\end{tabular}
\caption{In (a), we show the ratio of nontrivial information flows discovered using \toolname versus existing specifications. We show the ratio of nontrivial points-to edges discovered using (b) \toolname versus ground truth and (c) ground truth versus implementation. The ratios are sorted from highest to lowest for the 46 benchmark programs with nontrivial points-to edges. Note that some values exceeded the graph scale; in (a), the largest value is 9.3, and in (c), the largest value is 11.7.}
\label{fig:evalpt}
\end{figure*}

\section{Evaluation}
\label{sec:eval}

We implemented our specification inference algorithm in a tool called \toolname and evaluated its ability to infer points-to specifications for the Java standard library. First, we demonstrate the usefulness of the inferred specifications for a static information flow analysis, comparing to the existing, handwritten specifications, even though many of these specifications were written specifically for apps in our benchmark. Second, we evaluate the precision and recall of the inferred specifications by comparing to ground truth specifications; furthermore, we demonstrate the effectiveness of using specifications by showing that using ground truth specifications significantly decreases false positives compared to analyzing the actual library implementation. Finally, we analyze some of the choices we made when designing our algorithm.

\paragraph{{\bf\em Benchmark.}}

We base our evaluation on a benchmark of 46 Android apps, including a benchmark of 26 malicious and benign apps given to us by a major security company. The remaining apps were obtained as part of a DARPA program on detecting Android malware. Overall, our benchmark contains a mix of utility apps (e.g., flashlights, note taking apps, battery monitors, wallpaper apps, etc.) and Android games. We show the sizes of these apps in Jimple lines of code (Jimple is the intermediate representation used by Soot) in Figure~\ref{fig:sizes}. The malware in this benchmark consist primarily of apps that leak sensitive user information, including location, contacts, phone number, and SMS messages.

\paragraph{{\bf\em Information flow client.}}

We use \toolname to infer specifications for a client static information flow analysis for Android apps~\cite{fuchs2009scandroid,feng2014apposcopy,arzt2014flowdroid,bastani2015interactively}; in particular, it is based closely on~\cite{bastani2015specification}. This information flow client is specifically designed to find Android malware exhibiting malicious behaviors such as the ones in our benchmark. The client is built on Chord~\cite{naik2006effective} modified to use Soot~\cite{vallee1999soot} as a backend. It computes an Andersen-style, $1$-object-sensitive, flow- and path-insensitive points-to analysis. It then computes an explicit information flow analysis~\cite{sabelfeld2003language} between manually annotated sources and sinks, using the computed points-to sets to resolve flows through the heap. Sources include location, contacts, and device information (e.g., the return value of \texttt{getLocation}), and sinks include the Internet and SMS messages (e.g., the \texttt{text} parameter of \texttt{sendTextMessage}).

\paragraph{{\bf\em Existing specifications.}}

Our tool omits analyzing the Java standard library (version 7), and instead analyzes client code and user-provided code fragment specifications. Over the course of two years, we have handwritten several hundred code fragment specifications for our static information flow client for 90 classes in the Java standard library, including many written specifically for our benchmark of Android apps. In particular, these specifications were written as needed for the apps that we analyzed. The most time consuming aspect of developing the existing specifications was not writing them (they total a few thousand lines of code), but identifying which functions required specifications. Thus, they cover many fewer functions than the handwritten ground truth specifications described in Section~\ref{sec:groundtruth}, but are tailored to finding information flows for apps in our benchmark.

\paragraph{{\bf\em Evaluating inferred specifications.}}

To measure precision and recall of our inferred specifications, we begin by converting inferred specifications $\hat{S}$ to code fragment specifications using the algorithm described in Appendix~\ref{sec:pathanalysis}. Next, recall our observation in Section~\ref{sec:path} that the library implementation can be soundly and precisely represented by a regular set of path specifications $S_*$; we can similarly convert $S_*$ to a set of code fragment specifications. Then, we count a code fragment specification in $\hat{S}$ as a false positive if it does not appear in $S_*$, and similarly count a code fragment specification in $S_*$ as a false negative if it does not appear in $\hat{S}$. For code fragment specifications with multiple statements, we count each statement fractionally. For example, consider the sound and precise specification
\begin{small}
\begin{verbatim}
class Box { // ground truth specification
  Object f; // ghost field
  Object set(Object ob) {
    f = ob;
    return ob;
    return f; }
\end{verbatim}
\end{small}
for the \texttt{set} method in the \texttt{Box} class (note that the statement \texttt{return ob} is redundant, but it is generated by the algorithm described in Appendix~\ref{sec:pathanalysis}). Then, the specification
\begin{small}
\begin{verbatim}
Object set(Object ob) { // inferred specification
  return ob; }
\end{verbatim}
\end{small}
is missing two statements, so we count it as $2/3$ of a false negative. Each edge $z_i\dashrightarrow w_i$ in a path specification roughly corresponds to a single statement in the generated code fragments, so this heuristic intuitively counts false negative and false positive path specifications weighted by their length.

\paragraph{{\bf\em Evaluating computed relations.}}

To compare how using two different sets of specifications impacts the static analysis, we examine the ratio of the sizes of the computed relations. We preprocess these sizes in two ways: (i) we only consider relations between program variables, and (ii) we ignore \emph{trivial} relations that can be computed even when using empty specifications (i.e., all library functions are treated as no-ops with respect to heap effects). For example, for points-to edges, we use the metric $R_{\text{pt}}(S,S')=\frac{|\Pi(S)\setminus\Pi(\varnothing)|}{|\Pi(S')\setminus\Pi(\varnothing)|}$, where $\Pi(S)\subseteq\mathcal{V}\times\mathcal{O}$ are the points-to edges computed using specifications $S$, and $\Pi(\varnothing)\subseteq\mathcal{V}\times\mathcal{O}$ are the trivial edges.


\subsection{Comparison to Our Existing Specifications}

We compare the quality of the inferred specifications to our existing, handwritten specifications, in particular, aiming to improve our information flow client. We focus on inferring specifications for the commonly used packages \texttt{java.lang}, \texttt{java.util}, \texttt{java.io}, \texttt{java.nio}, \texttt{java.net}, and \texttt{java.math}; all the specifications we have manually written for the Java standard library are for classes in these packages. We note that the handwritten specifications for these packages are high-quality---we manually examined all of these handwritten specifications, and found that they were precise (i.e., there were no false positives compared to $S_*$).

\paragraph{{\bf\em Inferred specifications.}}

We used a total of 12 million random samples for phase one, which ran in 44.9 minutes. Phase two of our algorithm ran in 31.0 minutes; the initial FSA had 10,969 states, and the final FSA had 6,855 states. We compare our inferred specifications to the handwritten ones, i.e., measuring precision and recall compared to the handwritten specifications rather than to $S_*$. Most strikingly, \toolname infers 5$\times$ as many specifications as the existing, handwritten ones---\toolname infers specifications for 878 library functions, whereas handwritten specifications cover only 159 library functions. Furthermore, \toolname infers 89\% of the handwritten specifications.
We manually examined the 5 false negatives (i.e., handwritten specifications that \toolname fails to infer) for the Collections API. Each one is due to a false negative in the unit test synthesis. For example, the function \texttt{subList(int,int)} in the \texttt{List} class requires a call of the form \texttt{subList(0, 1)} to retrieve the first object in the list. Similarly, the function \texttt{set(int,Object)} in the \texttt{List} class requires an object to already be in the list or it raises an index out-of-bounds exception. The potential witnesses synthesized by \toolname fail to exercise the relevant behaviors in these instances. Finally, we manually examined more than 200 of the inferred specifications that were new; all of them were precise, which is strong evidence that our tool has very few if any false positives despite the heuristics we employ.

\paragraph{{\bf\em Information flows.}}

To show how \toolname can improve upon our existing handwritten specifications, we study the ratio $R_{\text{flow}}(S_{\text{atlas}},S_{\text{hand}})$ of information flows computed using \toolname versus using the existing, handwritten specifications. A higher ratio ($R_{\text{flow}}>1$) says that \toolname has higher recall, and a lower ratio ($R_{\text{flow}}<1$) says that handwritten specifications have higher recall. Figure~\ref{fig:evalpt} (a) shows $R(S_{\text{atlas}},S_{\text{hand}})$. Overall, \toolname finds 52\% more information flows compared to the handwritten specifications. The size of this gap is noteworthy because we have already analyzed these apps over the past few years---many of the existing specifications were written specifically for this benchmark.

Finally, for the subset of apps in our benchmark given to us by a major security company (27 of the 46) and for a subset of information sources and sinks, the security company provided us with ground truth information flows that they considered to be malicious. In particular, they consider 86.5\% of the information flows newly identified using the inferred specifications to be actual malicious behaviors.

\subsection{Comparison to Ground Truth}
\label{sec:groundtruth}

Since the existing handwritten specifications are incomplete, we additionally compare to the ground truth specifications $S_*$ to evaluate the precision and recall of our inferred specifications. Because of the manual effort required to write ground truth specifications, we do so only for the 12 classes in the Java Collections API that are most frequently used by our benchmark (98.5\% of calls to the Collections API target these 12 classes). We focus on the Java Collections API (i.e., classes that implement the \texttt{Collection} or \texttt{Map} interfaces), because it requires by far the most complex points-to specifications.

\paragraph{{\bf\em Inferred specifications.}}

We examine the top 50 most frequently called functions in our benchmark (in total, accounting for 95\% of the function calls).
The recall of our algorithm is 97\% (i.e., we inferred the ground truth specification for 97\% of the 50 functions) and the precision is 100\% (i.e., each specification is as precise as the ground truth specification). The false negatives occured for the same reasons as the false negatives discussed in the previous section.

\paragraph{{\bf\em Points-to sets.}}

To show the quality of the specifications inferred by \toolname, we study the ratio $R_{\text{pt}}(S_{\text{atlas}},S_*)$ of using specifications inferred by \toolname to using ground truth specifications. We found that using \toolname does not compute a single false positive points-to edge compared to using ground truth specifications, i.e., the precision of \toolname is 100\%. Thus, $1-R_{\text{pt}}(S_{\text{atlas}},S_*)$ is the rate of false negative points-to edges when using \toolname. Figure~\ref{fig:evalpt} (b) shows $R_{\text{pt}}(S_{\text{atlas}},S_*)$ for each app in our benchmark, sorted by magnitude. This ratio is 1.0 for almost half of the programs, i.e., for almost half the programs, there are no false negatives. The median recall is 99.0\%, and the average recall is 75.8\%.

\paragraph{{\bf\em Benefits of using specifications.}}

We show that using specifications can greatly improve the precision and soundness of a static analysis. In particular, we compare the ground truth specifications to the library implementation, i.e., the class files comprising the actual implementation of the Collections API (developed by Oracle). We compute the ratio $R_{\text{pt}}(S_{\text{impl}},S_*)$ of analyzing the library implementation $S_{\text{impl}}$ to analyzing the ground truth specifications $S_*$. This ratio measures the number of false positives due to analyzing the library implementation instead of using ground truth specifications, since every points-to edge computed using the implementation but not the ground truth specifications is a false positive. Figure~\ref{fig:evalpt} (c) shows this ratio $R_{\text{pt}}(S_{\text{impl}},S_*)$. For a third of programs, the false positive rate is more than 100\% (i.e., when $R_{\text{pt}}\ge2$), and for four programs, the false positive rate is more than 300\% (i.e., $R_{\text{pt}}\ge4$). The average false positive rate is 115.2\%, and the median is 62.1\%. Furthermore, for two of the programs, there are false negatives (i.e., $R_{\text{pt}}<1$) due to unanalyzable calls to native code.

\subsection{Design Choices}
\label{sec:evalspec}

\begin{figure*}
\begin{minipage}{0.45\textwidth}
\footnotesize
\begin{verbatim}
class StrangeBox { // library
  Object f;
  void set(Object ob) {
    f = ob;
    f = null; }
  Object get() { return f; } }
\end{verbatim}
\end{minipage}
\begin{minipage}{0.45\textwidth}
\footnotesize
\begin{verbatim}
boolean test() { // program
  Object in = new Object(); // o_in
  Box box = new Box(); // o_box
  Object out;
  new Thread(() -> box.set(in)).execute();
  new Thread(() -> out = box.get()).execute();
  return in == out; }
\end{verbatim}
\end{minipage}
\caption{Implementation of the library methods \texttt{set}, \texttt{get}, and \texttt{clone} in the \texttt{StrangeBox} class (left), and an example of a concurrent program using these functions (right).}
\label{fig:strangeexample}
\end{figure*}

Finally, we compare the performance of different design choices for our specification inference algorithm; in particular, we infer specifications for 733 library functions in the Java Collections API using different design choices.

\paragraph{{\bf\em Positive examples: random sampling vs. MCTS.}}

We sampled 2 million candidate path specifications using each algorithm. Random sampling found 3,124 positive examples, whereas MCTS found 10,153.

\paragraph{{\bf\em Object initialization: null vs. instantiation.}}

Each of the 11,613 positive examples passed the unit test constructed using instantiation, but only 7,721 passed when using null initialization, i.e., instantiation finds 50\% more specifications. As discussed above, even with this heuristic, we estimate that the false positive rate of our tool is zero.

\subsection{Discussion}

As we have shown, using ground truth substantially improves precision and soundness compared to analyzing the library implementation. For the one-time cost of writing specifications, we can eliminate imprecision due to deep call hierarchies and unsoundness due to native code, reflection, etc. Such an approach is already used in production static analysis systems to handle hard-to-analyze code~\cite{facebook2017}. However, manually writing specifications that are complete or close to complete is impractical---ground truth specifications for just 12 classes took more than a week to write and contain 1,731 lines of code, but there are more than 4,000 classes in the Java standard library. Typically, manual effort is focused on writing specifications for the most commonly used functions, but this approach leaves a long tail of missing specifications~\cite{bastani2015specification}. Moreover, handwritten specifications can be error prone~\cite{heule2016stratified}, and libraries grow over time---the Java 9 standard library contains more than 2,000 new classes---so specifications must be maintained over time.

We have shown that \toolname automatically covers an order of magnitude more of the Java Collections API compared to the existing, handwritten specifications, and is furthermore very close to ground truth. In addition, we have shown that \toolname can substantially improve recall on an information flow client compared to the handwritten specifications, despite the fact that we have already written specifications specifically for apps in this benchmark. Thus, \toolname substantially improves the practicality of using specifications to model hard-to-analyze code in static analysis. We believe that the remaining gap can be bridged by the human analyst, e.g., by manually inspecting important specifications or by using interactive specification inference~\cite{zhu2013automated,bastani2015specification}.

\section{Discussion}
\label{sec:discussion}

\paragraph{{\bf\em Sources of imprecision.}}

We briefly summarize the two potential sources of imprecision in our specification inference algorithm. First, in the second phase of our algorithm, imprecision can be introduced when we merge FSA states. In particular, such a merge can add infinitely many new path specifications to our set of inferred specifications; we only check precision for specifications up to a bounded length. Thus, newly added specifications that are longer than this bound may be imprecise. In our evaluation, we observe that all the added specifications are precise. Second, in both phases of our analysis, the unit test synthesis algorithm uses a heuristic where it initializes all parameters to non-null values. Thus, the unit test synthesized for a given specification may pass even if the specification is imprecise. In our evaluation, we show that this approximation helps us find many more specifications without introducing any imprecision.

\paragraph{{\bf\em Sources of unsoundness.}}

We briefly summarize the potential sources of unsoundness in our specification inference algorithm, i.e., why it may miss correct specifications. At a high level, there are three sources: (i) the restriction to regular languages, (ii) incompleteness in the search (both in the random samples in phase one, and in the language inference algorithm), and (iii) shortcomings in the unit test synthesis algorithm, i.e., the unit test synthesized for a given specification fails even though the specification is precise.

In our evaluation, we found that unsoundness was entirely due to reason (iii). We believe there are three reasons why a synthesized unit test may fail to pass: (a) heuristics used for scheduling, (b) heuristics used to initialize variables, and (c) concurrency. First, in the scheduling step, there are multiple possible schedules of the statements in the unit test, and our synthesis algorithm uses heuristics to choose a good one (i.e., one that is likely to be a witness). Second, in the initialization step, there are multiple possible ways to initialize variables in the unit test; again, our algorithm uses heuristics to choose a good one. As described in Section~\ref{sec:eval}, in our evaluation, all incorrectly rejected specifications were due to this second reason (e.g., this reason explains our algorithm failed to infer the specifications for \texttt{subList} and \texttt{set}).

The third possible reason is due to concurrency. So far, we have implicitly assumed that code is executed sequentially. In principle, because we are using a flow-insensitive points-to analysis, path specifications are sound with respect to concurrency. However, it may be possible that a unit test must include concurrent code to avoid incorrectly rejecting a specification. For example, consider the specifications for the class \texttt{StrangeBox} and the test program using this class shown in Figure~\ref{fig:strangeexample}. Because our static analysis is flow insensitive, it soundly determines that the argument of \texttt{set} may be aliased with the return value of \texttt{get}. However, for any sequential client code, the statement \texttt{f = ob} in the \texttt{set} method has no effects; thus, such code can never observe that \texttt{ob} and $r_{\texttt{get}}$ may be aliased. In particular, the unit test synthesized by our algorithm for the candidate path specification
\begin{align*}
  s_{\text{strange}}=\texttt{ob}\dashrightarrow\texttt{this}_{\texttt{set}}\rightarrow\texttt{this}_{\texttt{get}}\dashrightarrow r_{\texttt{get}}.
\end{align*}
would fail (thereby rejecting $s_{\text{strange}}$), even though $s_{\text{strange}}$ is precise. Still, a witness exists, i.e., the unit test shown in Figure~\ref{fig:strangeexample}, which executes \texttt{set} and \texttt{get} concurrently.

\paragraph{{\bf\em Potential clients.}}

We believe that the points-to summaries inferred by our algorithm may be useful for a number of clients beyond information flow analysis. In particular, our approach is applicable to any static analysis where soundness is highly desirable, but where state-of-the-art tools nevertheless sacrifice soundness in favor of reducing false positives. This criterion is true for nearly all static analyses designed to find bugs and security vulnerabilities, which frequently depend heavily on points-to analysis, information flow analysis, and taint analysis. Even for the security-critical task of finding malware, existing static analyses are typically unsound with respect to features such as reflection, dynamically loaded code, and native code~\cite{arzt2014flowdroid,feng2014apposcopy}. Using specifications not only improves precision, but also reduces unsoundness due to the use of these features in large libraries. If necessary, the human security analyst can always add specifications to improve the static analysis---for this use case, ATLAS substantially reduces human workload, e.g., in our evaluation, ATLAS inferred 92\% of handwritten specifications for the Java Collections API. On the other hand, our approach is not suitable for static analyses where soundness is crucial, such as compiler optimizations.

\section{Related Work}
\label{sec:related}

\paragraph{{\bf\em Inferring specifications for library code.}}

Techniques have been proposed for mining specifications for library code from executions, e.g., taint specifications (i.e., whether taint flows from the argument to the return value)~\cite{clapp2015modelgen}, functional specifications of library functions~\cite{heule2015mimic}, specifications for x86 instructions~\cite{heule2016stratified}, and specifications for callback control flow~\cite{jeon2016synthesizing}. In contrast, points-to specifications that span multiple functions are more complex properties.
One approach is to infer points-to specifications using data gathered from deployed instrumented apps~\cite{bastani2017eventually}. In contrast, our algorithm actively synthesizes unit tests that exercise the library code and requires no instrumentation of deployed apps.
Another approach is to interact with a human analyst to infer specifications~\cite{zhu2013automated,bastani2015specification,albarghouthi2016maximal}. These approach guarantee soundness, but still require substantial human effort, e.g., in~\cite{bastani2015specification}, the analyst may need to write more than a dozen points-to specifications to analyze a single app. Finally,~\cite{ali2013averroes} uses a static approach to infer callgraph specifications for library code.

\paragraph{{\bf\em Inferring program properties.}}

There has been work inferring program invariants from executions~\cite{nimmer2002automatic}, including approaches using machine learning~\cite{sharma2012interpolants,sharma2013data,sharma2014invariant}. The most closely related work is~\cite{bastani2017synthesizing}, which uses an active learning strategy to infer program input grammars for blackbox code. In contrast, our goal is to infer points-to specifications for library code. There has also been work on specifications encoding desired properties of client programs (rather than encoding behaviors of the library code), both using dynamic analysis~\cite{kremenek2006uncertainty,ramanathan2007static,shoham2008static,livshits2009merlin,beckman2011probabilistic,bastani2015interactively} and using static analysis ~\cite{ammons2002mining,yang2006perracotta}.

\paragraph{{\bf\em Static points-to analysis.}}

There is a large literature on static points-to analysis~\cite{shivers1991control,andersen1994program,wilson1995efficient,fahndrich1998partial,milanova2002parameterized}, including formulations based on set-constraints and context-free language reachability~\cite{reps1998program,kodumal2004set,kodumal2005banshee,sridharan2005demand}. Recent work has focused on improving context-sensitivity~\cite{whaley2004cloning,sridharan2006refinement,liang2011scaling,zhang2014abstraction,smaragdakis2014introspective}. Using specifications in conjunction with these analyses can improve precision, scalability, and even soundness.
One alternative is to use demand driven static analyses to avoid analyzing the entire library code~\cite{sridharan2005demand}; however, these approaches are not designed to work with missing code, and furthermore do not provide much benefit for demanding clients that require analyzing a substantial fraction of the library code.

\section{Conclusion}

Specifications summarizing the points-to effects of library code can be used to increase precision, recall, and scalability of running a static points-to analysis on any client code. By automatically inferring such specifications, \toolname fully automatically achieves all of these benefits without the typical time-consuming and error-prone process of writing specifications by hand. We believe that \toolname is an important step towards improving the usability of static analysis.

\begin{acks}

{\small This material is based on research sponsored by DARPA under agreement number FA84750-14-2-0006. The U.S. Government is authorized to reproduce and distribute reprints for Governmental purposes notwithstanding any copyright notation thereon. The views and conclusions herein are those of the authors and should not be interpreted as necessarily representing the official policies or endorsements either expressed or implied of DARPA or the U.S. Government. This work was also supported by NSF grant CCF-1160904 and a Google Fellowship.}

\end{acks}


\bibliography{paper}

\clearpage
\appendix
\section{Static Points-To Analysis with Regular Sets of Path Specifications}
\label{sec:pathanalysis}

In this section, we describe how to run our static points-to analysis in conjunction with a possibly infinite regular set $S$ of path specifications (assumed to be represented as an FSA, i.e., $S=\mathcal{L}(\hat{M})$). In particular, our static analysis converts $S$ to a set $\tilde{S}$ of \emph{code fragment specifications}, which are replacements for the library code that have the same points-to effects as encoded by $S$.

Given path specifications $S$, our static analysis constructs \emph{equivalent} code fragment specifications $\tilde{S}$, i.e., $\overline{G}(P,S)=\overline{G}(P,\tilde{S})$. In other words, $\tilde{S}$ has the same semantics as $S$ with respect to our static points-to analysis. One detail in our definition of equivalence is that $\overline{G}(P,\tilde{S})$ may contain additional vertices corresponding to variables and abstract objects in the code fragment specifications; we omit these extra vertices and their relations at the end of the static analysis.

\subsection{Converting a Single Path Specification}

For intuition, we begin by describing how to convert a single path specification
\begin{align*}
s=(z_1\dashrightarrow w_1\to...\to z_k\dashrightarrow w_k)
\end{align*}
into an equivalent set of code fragment specifications, where $A_i=\Alias$ for each $i$ and $z_1$ is a parameter. Let the code fragment specifications $\tilde{S}$ corresponding to $s$ be:
\begin{align*}
m_1&=\{w_1.f_1\gets z_1\} \\
m_2&=\{t_2\gets z_2.f_1,\sss w_2.f_2\gets t_2\} \\
&... \\
m_k&=\{w_k\gets z_k.f_{k-1}\},
\end{align*}
where $f_1,...f_{k-1}\in\mathcal{F}$ are fresh fields and $t_2,...,t_{k-1}$ are fresh variables. Then:
\begin{proposition}
\label{prop:codegensingle}
\rm
We have $\overline{G}(P,\tilde{S})=\overline{G}(P,\{s\})\cup\overline{G}'(P,\tilde{S})$, where $\overline{G}'(P,\tilde{S})$ consists of the edges in $\overline{G}(P,\tilde{S})$ that refer to vertices corresponding to variables and abstract objects in $\tilde{S}$.
\end{proposition}
\begin{proof}
(sketch) First, we show that $\overline{G}(P,\{s\}\subseteq\overline{G}(P,\tilde{S})$. Suppose that the premise of $s$ holds, i.e., $z_i\xrightarrow{A_i}w_{i+1}\in\overline{G}$ for each $i$. Then, the static analysis computes $z_1\xrightarrow{\Transfer}w_k\in\overline{G}(P,\{s\})$; we need to show that $z_1\xrightarrow{\Transfer}w_k\in\overline{G}(P,\tilde{S})$ as well. Note that we have
\begin{align*}
&z_1\xrightarrow{\Store[f_1]}w_1\xrightarrow{\Alias}z_2\xrightarrow{\Load[f_1]}t_2\in\overline{G}(P,\tilde{S}) \\
&t_2\xrightarrow{\Store[f_2]}w_2\xrightarrow{\Alias}z_3\xrightarrow{\Load[f_2]}t_3\in\overline{G}(P,\tilde{S}) \\
&... \\
&t_{k-1}\xrightarrow{\Store[f_{k-1}]}w_{k-1}\xrightarrow{\Alias}z_k\xrightarrow{\Load[f_{k-1}]}w_k\in\overline{G}(P,\tilde{S}).
\end{align*}
By induction, the static analysis computes $z_1\xrightarrow{\Transfer}t_i\in\overline{G}(P,\tilde{S})$ for each $i\in[k-1]$. Thus, the static analysis computes $z_1\xrightarrow{\Transfer}w_k\in\overline{G}(P,\tilde{S})$, as claimed.

\begin{figure}
\centering
\small
  \[\text{(initial parameter)}~\dfrac{q_{\text{init}}\xrightarrow{z}q\xrightarrow{w}r\in\hat{M},\hspace{0.1in}z=p_m,\hspace{0.1in}w\in\{p_m,r_m\}}{w.f_r\gets z\in m}\]
  \[\text{(initial return)}~\dfrac{q_{\text{init}}\xrightarrow{z}q\xrightarrow{w}r\in\hat{M},\hspace{0.1in}z=r_m,\hspace{0.1in}w\in\{p_m,r_m\}}{t\gets X(),~z\gets t,~w.f_r\gets t\in m}\]
  \[\text{(final parameter)}~\dfrac{p\xrightarrow{z}q\xrightarrow{w}q_{\text{fin}}\in\hat{M},\hspace{0.1in}z=p_m,\hspace{0.1in}w=r_m}{w\gets z.f_p\in m}\]
  \[\text{(final return)}~\dfrac{p\xrightarrow{z}q\xrightarrow{w}q_{\text{fin}}\in\hat{M},\hspace{0.1in}z=r_m,\hspace{0.1in}w=r_m}{t\gets X(),~z.f_p\gets t,~w\gets t\in m}\]
  \[\text{(}A_i=\Alias\text{)}~\dfrac{p\xrightarrow{z}q\xrightarrow{w}r\in\hat{M},\hspace{0.1in}z=p_m,\hspace{0.1in}w=p_m}{t\gets z.f_p,~w.f_r\gets t\in m}\]
  \[\text{(}A_i=\Transfer\text{)}~\dfrac{p\xrightarrow{z}q\xrightarrow{w}r\in\hat{M},\hspace{0.1in}z=p_m,\hspace{0.1in}w=r_m}{w\text X(),~t\gets z.f_p,~w.f_r\gets t\in m}\]
  \[\text{(}A_i=\overline{\Transfer}\text{)}~\dfrac{p\xrightarrow{z}q\xrightarrow{w}r\in\hat{M},\hspace{0.1in}\{z,w\}\subseteq\{p_m,r_m\}}{z\gets X(),~t\gets w.f_r,~z.f_p\gets t\in m}\]
  \[\text{(initial final)}~\dfrac{q_{\text{init}}\xrightarrow{z}q\xrightarrow{w}q_{\text{fin}}\in\hat{M},\hspace{0.1in}\{z,w\}\subseteq\{p_m,r_m\}}{w\gets z\in m}\]
\caption{Rules for generating code fragment specifications from path specifications defined by a finite state automaton $\hat{M}=(Q,\mathcal{V}_{\text{path}},\delta,q_{\text{init}},Q_{\text{fin}})$, where for simplicity we assume $\hat{M}$ has a single accept state $q_{\text{fin}}$.}
\label{fig:codegen}
\end{figure}

\begin{figure*}
\footnotesize
\centering
\begin{tabular}{lll}
\hline
\\
\multicolumn{1}{c}{{\bf Candidate (Regular Expression)}} & \multicolumn{1}{c}{{\bf Candidate (Finite State Automaton)}} & \multicolumn{1}{c}{{\bf Code Fragments}} \\\\
\hline
\\
\hspace{0.025in}
\begin{minipage}{2.0in}
$\begin{array}{l}
\texttt{ob}\dashrightarrow\texttt{this}_{\texttt{set}}\to\texttt{this}_{\texttt{get}}\dashrightarrow r_{\texttt{get}}
\end{array}$
\end{minipage}
&
\begin{minipage}{2.8in}
$\begin{tikzcd}[column sep=3.5em,row sep=2em]
q_{\text{init}} \arrow[r,"{\texttt{ob}}"] 
& q_1 \arrow[r,"{\texttt{this}_{\texttt{set}}}"]
& q_{\texttt{f}} \arrow[r,"{\texttt{this}_{\texttt{get}}}"]
& q_2 \arrow[r,"{r_{\texttt{get}}}"]
& q_{\text{fin}}
\end{tikzcd}$
\end{minipage}
&
\begin{minipage}{1.6in}
\begin{verbatim}
void set(Object ob) { f = ob; }
Object get() { return f; }
\end{verbatim}
\end{minipage}
\hspace{0.025in}
\\\\
\hline
\\
\hspace{0.025in}
\begin{minipage}{2.0in}
$\begin{array}{ll}
\texttt{ob}\dashrightarrow\texttt{this}_{\texttt{set}}&\hspace{-0.1in}\blue{\big(}\to\texttt{this}_{\texttt{clone}}\dashrightarrow r_{\texttt{clone}}\blue{\big)^*} \\
&\hspace{-0.1in}\to\texttt{this}_{\texttt{get}}\dashrightarrow r_{\texttt{get}}
\end{array}$
\end{minipage}
&
\begin{minipage}{2.8in}
$\begin{tikzcd}[column sep=3.5em,row sep=2em]
q_{\text{init}} \arrow[r,"{\texttt{ob}}"] 
& q_1 \arrow[r,"{\texttt{this}_{\texttt{set}}}"]
& q_{\texttt{f}} \arrow[r,"{\texttt{this}_{\texttt{get}}}"] \arrow[bend right]{d}[xshift=-11ex, yshift=-2ex]{\texttt{this}_{\texttt{clone}}}
& q_2 \arrow[r,"{r_{\texttt{get}}}"]
& q_{\text{fin}} \\
&& q_3 \arrow[bend right]{u}[xshift = 8ex, yshift=-2ex]{r_{\texttt{clone}}}
\end{tikzcd}$
\end{minipage}
&
\begin{minipage}{1.6in}
\begin{verbatim}
void set(Object ob) { f = ob; }
Object get() { return f; }
Box clone() {
  Box b = new Box(); // ~o_clone
  b.f = f;
  return b; } }
\end{verbatim}
\end{minipage}
\hspace{0.025in}
\\\\
\hline
\\
\hspace{0.025in}
\begin{minipage}{2.0in}
$\begin{array}{ll}
\multicolumn{2}{l}{\texttt{ob}\dashrightarrow\texttt{this}_{\texttt{set}}\to\texttt{this}_{\texttt{get}}\dashrightarrow r_{\texttt{get}}} \vspace{0.1in} \\
\blue{+}~\texttt{ob}\dashrightarrow\texttt{this}_{\texttt{set}}&\hspace{-0.1in}\to\texttt{this}_{\texttt{clone}}\dashrightarrow r_{\texttt{clone}} \\
&\hspace{-0.1in}\to\texttt{this}_{\texttt{get}}\to r_{\texttt{get}} \vspace{0.1in} \\
\blue{+}~\texttt{ob}\dashrightarrow\texttt{this}_{\texttt{set}}&\hspace{-0.1in}\to\texttt{this}_{\texttt{clone}}\dashrightarrow r_{\texttt{clone}} \\
&\hspace{-0.1in}\to\texttt{this}_{\texttt{clone}}\dashrightarrow r_{\texttt{clone}} \\
&\hspace{-0.1in}\to\texttt{this}_{\texttt{get}}\dashrightarrow r_{\texttt{get}}
\end{array}$
\end{minipage}
&
\begin{minipage}{2.8in}
$\begin{tikzcd}[column sep=3.5em,row sep=1.6em]
q_{\text{init}} \arrow[r,"{\texttt{ob}}"] 
& q_1 \arrow[r,"{\texttt{this}_{\texttt{set}}}"]
& q_{\texttt{f}} \arrow[r,"{\texttt{this}_{\texttt{get}}}"] \arrow[d]{u}[xshift=-12ex]{\texttt{this}_{\texttt{clone}}}
& q_2 \arrow[r,"{r_{\texttt{get}}}"]
& q_{\text{fin}} \\
&& q_3 \arrow[d]{u}[xshift=-8ex]{r_{\texttt{clone}}}
&
&
\\
&& q_{\texttt{g}} \arrow[bend right]{uur}[xshift=6ex,yshift=-5ex]{\texttt{this}_{\texttt{get}}} \arrow[d]{u}[xshift=-12ex]{\texttt{this}_{\texttt{clone}}}
\\
&& q_4 \arrow[d]{u}[xshift=-8ex]{r_{\texttt{clone}}} \\
&& q_{\texttt{h}} \arrow[bend right=60]{uuuur}[xshift=7ex,yshift=-7ex]{\texttt{this}_{\texttt{get}}}
\end{tikzcd}$
\end{minipage}
&
\begin{minipage}{1.6in}
\begin{verbatim}
void set(Object ob) { f = ob; }
Object get() {
  return f;
  return g;
  return h; }
Box clone() {
  Box b = new Box(); // ~o_clone
  b.g = f;
  b.h = g;
  return b; } }
\end{verbatim}
\end{minipage}
\hspace{0.025in}
\\\\
\hline
\end{tabular}
\caption{
  Examples of candidate code fragment specifications (left column), and the equivalent path specifications as a regular expression (middle column) and as a finite state automaton (right column).
}
\label{fig:pathexamples}
\end{figure*}

Next, we show the converse, i.e., that $\overline{G}(P,\tilde{S})\subseteq\overline{G}(P,S)\cup\overline{G}'(P,\tilde{S})$. First, note that the only production with $\Store[f]$ is
\begin{align*}
\Transfer\to\Transfer~\Store[f]~\Alias~\Load[f].
\end{align*}
Since each $f_i$ is a fresh field, there is only one edge labeled $\Store[f_i]$ and only one edge labeled $\Load[f_i]$. Thus, this production can only be triggered if (i) $z_i\xrightarrow{\Alias}w_i\in\overline{G}(P,\tilde{S})$, and (ii) for some vertex $x$, $x\xrightarrow{\Transfer}t_i\in\overline{G}(P,\tilde{S})$. If triggered, the static analysis adds an edge $x\xrightarrow{\Transfer}t_{i+1}$ to $\overline{G}(P,\tilde{S})$. For $i=1$, the only vertices $x$ satisfying the second condition are $x=z_1$ and $x=t_1$. By induction, if $w_i\xrightarrow{\Alias}z_{i+1}\in\overline{G}(P,\tilde{S})$ for each $i$, we have
\begin{align*}
z_1&\xrightarrow{\Transfer}t_i\in\overline{G}(P,\tilde{S}) \\
t_j&\xrightarrow{\Transfer}t_i\in\overline{G}(P,\tilde{S})
\end{align*}
for each $j\le i$. None of the $t_i$ are part of an $\Assign$ edge except $t_1$ and $t_k$; for the latter, the production $\Transfer\to\Transfer~\Assign$ triggers and we get $z_1\xrightarrow{\Transfer}w_k\in\overline{G}(P,\tilde{S})$. This edge is the only one in $\overline{G}(P,\tilde{S})$ that does not refer to vertices extracted from the code fragments, so the claim follows.
\end{proof}

\subsection{Converting a Regular Set of Path Specifications}

Our construction generalizes straightforwardly to constructing code fragment specifications from $\hat{M}$. For each state $q\in Q$, we introduce a fresh field $f_q\in\mathcal{F}$. Intuitively, transitions into $q$ correspond to stores into $f_q$, and transitions coming out of $q$ correspond to loads into $f_q$. In particular, we include statements in $m$ according to the rules in Figure~\ref{fig:codegen}.

The following guarantee follows similarly to the proof of Proposition~\ref{prop:codegensingle}:
\begin{proposition}
\label{PROP:CODEGEN}
\rm
We have $\overline{G}(P,\tilde{S})=\overline{G}(P,S)\cup\overline{G}'(P,\tilde{S})$, where $\overline{G}'(P,\tilde{S})$ is defined as before.
\end{proposition}

In Figure~\ref{fig:pathexamples}, we show examples of path specifications (first column), the corresponding FSA (middle column), and the generated code fragment specifications. For example, in the second line, the transitions
\begin{align*}
q_{\text{init}}\xrightarrow{\texttt{ob}}q_1\xrightarrow{\texttt{this}_{\texttt{set}}}q_2\xrightarrow{\texttt{this}_{\texttt{get}}}q_3\xrightarrow{r_{\texttt{get}}}q_{\text{fin}}
\end{align*}
generate the specifications for \texttt{set} (the first two transitions, with field $\texttt{f}=f_{q_2}$) and \texttt{get} (the last two transitions), and the self-loop
\begin{align*}
q_2\xrightarrow{\texttt{this}_{\texttt{clone}}}q_6\xrightarrow{r_{\texttt{clone}}}q_2
\end{align*}
generates the specification for \texttt{clone}.

\section{Unit Test Synthesis Algorithm}
\label{sec:testsynthesisappendix}

In this section, we describe our algorithm for synthesizing a unit test to check correctness of a candidate path specification. As a running example, Figure~\ref{fig:testsynthesisappendix} shows how our unit test synthesis algorithm synthesizes a unit test for the candidate path specification $s_{\text{box}}$ for the \texttt{Box} class. In this example, the synthesized unit test contains exactly the external edges in the candidate's premise:
\begin{align*}
\texttt{this}_{\texttt{set}}\xrightarrow{\Alias}\texttt{this}_{\texttt{clone}},\sss r_{\texttt{clone}}\xrightarrow{\Transfer}\texttt{this}_{\texttt{set}}.
\end{align*}
Upon executing this unit test, the candidate's conclusion
\begin{align*}
\texttt{in}\xrightarrow{\Transfer}\texttt{out}
\end{align*}
holds dynamically. Therefore, this unit test witnesses the correctness of the given candidate.

Our algorithm first constructs a \emph{skeleton} containing a call to each function in the specification. Then, it (i) fills in \emph{holes} with variable names, (ii) initializes variables, and (iii) orders (or \emph{schedules}) statements. The last step also adds a statement returning whether the candidate's conclusion holds.

There are certain constraints on the choices that ensure that the synthesized unit test is a valid witness. Even with these constraints, a number of additional choices remain. Each choice produces a valid unit test, but some of these unit tests may not pass even if the candidate specification is correct. We describe the choices made by our algorithm, which empirically finds almost all correct candidate specifications.

\subsection{Skeleton Construction}

To witness correctness of the candidate path specification, the synthesized unit test must exhibit \emph{exactly} the external edges in its premise. In particular, the unit test must include a call to each function in the candidate. Our algorithm constructs a \emph{skeleton} consisting of these calls, for example, the skeleton on the second step of Figure~\ref{fig:testsynthesisappendix}. A symbol \texttt{??}, called a \emph{hole}, is included for each parameter and return value of each function call, and must be filled in with a variable name.

\subsection{Filling Holes}

The external edges in the candidate specification impose constraints on the arguments that should be used in each function call. In particular, the synthesized unit test must exhibit every behavior encoded by the external edges in the candidate specification:
\begin{itemize}
\item {\bf Alias:} For an aliasing edge $p_{m_i}\xrightarrow{\Alias}p_{m_{i+1}}$, the algorithm has to ensure that the arguments $p_{m_i}$ (passed to $m_i$) and $p_{m_{i+1}}$ (passed to $m_{i+1}$) are aliased.
\item {\bf Transfer:} For a transfer edge $r_{m_i}\xrightarrow{\Transfer}p_{m_{i+1}}$, the algorithm has to use the return value of $m_i$ as the argument passed to $m_{i+1}$ (and similarly for backwards transfer edges $p_{m_i}\xrightarrow{\overline{\Transfer}}r_{m_{i+1}}$).
\end{itemize}
For example, the holes in the skeleton in Figure~\ref{fig:testsynthesisappendix} are filled so that the following premises are satisfied:
\begin{align*}
\texttt{this}_{\texttt{set}}\xrightarrow{\Alias}\texttt{this}_{\texttt{clone}},\sss r_{\texttt{clone}}\xrightarrow{\Transfer}\texttt{this}_{\texttt{get}}.
\end{align*}

One issue is that internal edges may be self-loops, in which case more than two parameters may need to be aliased. For example, consider the following candidate:
\begin{align*}
\texttt{ob}\dashrightarrow\texttt{this}_{\texttt{set}}&\rightarrow\texttt{this}_{\texttt{clone}}\dashrightarrow\texttt{this}_{\texttt{id}} \\
\stepcounter{equation}\tag{\theequation}\label{eq:altspec}&\rightarrow\texttt{this}_{\texttt{get}}\dashrightarrow r_{\texttt{get}}.
\end{align*}
For the unit test for this candidate, the three calls to \texttt{set}, \texttt{clone}, and \texttt{set} must all share the same receiver:
\begin{small}
\begin{verbatim}
box.set(in);
Box boxClone = box.clone();
Object out = box.get();
\end{verbatim}
\end{small}

Our algorithm partitions the holes into subsets that must be aliased---since aliasing is a transitive relation, every hole in a subset has to be aliased with every other hole in that subset. To do so, the algorithm constructs an undirected graph where the vertices are the holes, and an edge $(h,h')\in E$ connects two holes $h$ and $h'$ in the following cases:
\begin{itemize}
\item There is an external edge $w_{m_i}\rightarrow z_{m_{i+1}}$ in the candidate specification, where $h$ is the hole corresponding to $w_{m_i}$ and $h'$ is the hole corresponding to $z_{m_{i+1}}$.
\item There is an internal edge $p_{m_i}\dashrightarrow p_{m_i}$ in the candidate specification, where $h$ is the hole corresponding to the $p_{m_i}$ on the left-hand side and $h'$ is the hole corresponding to the $p_{m_i}$ on the right-hand side.
\end{itemize}
Then, our algorithm computes the connected components in this graph. For each connected component, the algorithm chooses a fresh variable name, and each hole in that connected component is filled with this variable name.

For example, for the candidate in Figure~\ref{fig:testsynthesisappendix}, our algorithm computes the following partitions:
\begin{align*}
\{\texttt{ob}\},~\{\texttt{this}_{\texttt{set}},\texttt{this}_{\texttt{clone}}\},~\{r_{\texttt{clone}},\texttt{this}_{\texttt{get}}\},~\{r_{\texttt{get}}\},
\end{align*}
and fills the corresponding holes with the variables names
\begin{align*}
\texttt{in},~\texttt{box},~\texttt{boxClone},~\texttt{out},
\end{align*}
respectively. Similarly, for (\ref{eq:altspec}), we compute partitions
\begin{align*}
\{\texttt{ob}\},~\{\texttt{this}_{\texttt{set}},\texttt{this}_{\texttt{clone}},\texttt{this}_{\texttt{get}}\},~\{r_{\texttt{clone}}\},~\{r_{\texttt{get}}\}.
\end{align*}
The variable names are the same as those chosen in Figure~\ref{fig:testsynthesisappendix}.

\begin{figure}
\small
\begin{minipage}{0.4\textwidth}
$\begin{array}{rl}
\texttt{ob}\dashrightarrow\texttt{this}_{\texttt{set}}\hspace{-0.1in}&\rightarrow\texttt{this}_{\texttt{clone}}\dashrightarrow r_{\texttt{clone}} \\
&\rightarrow\texttt{this}_{\texttt{get}}\dashrightarrow r_{\texttt{get}}
\end{array}$
\end{minipage}
\begin{minipage}{0.5\textwidth}
\centering
\begin{tabular}{ll}
\hline \vspace{-0.1in} \\
$\begin{array}{l}\text{skeleton}\end{array}$ &
$\begin{array}{l}
\textttb{??.set(??);}\\
\textttb{?? = ??.clone();}\\
\textttb{?? = ??.get(??);}\\
\end{array}$
\vspace{0.03in} \\
\hline \vspace{-0.1in} \\
$\begin{array}{l}\text{fill holes}\end{array}$ &
$\begin{array}{l}
\textttb{box}\texttt{.set(}\textttb{in}\texttt{);}\\
\textttb{Box boxClone}\texttt{ = }\textttb{box}\texttt{.clone();}\\
\textttb{Object out}\texttt{ = }\textttb{boxClone}\texttt{.get();}
\end{array}$
\vspace{0.03in} \\
\hline \vspace{-0.1in} \\
$\begin{array}{l}\text{initialization}\\\text{\& scheduling}\end{array}$ &
$\begin{array}{l}
\textttb{Object in = new Object();}\\
\textttb{Box box = new Box()}\\
\texttt{box.set(in);}\\
\texttt{Box boxClone = box.clone();}\\
\texttt{Object out = boxClone.get();}\\
\textttb{return in == out};\\
\end{array}$
\vspace{0.03in} \\
\hline
\end{tabular}
\end{minipage}
\caption{Steps in the unit test synthesis algorithm (right) for a candidate path specification for \texttt{Box} (left). Code added at each step is highlighted in blue. Scheduling is shown in the same line as initialization---it chooses the final order of the statements. This figure is a duplicate of Figure~\ref{fig:testsynthesis}, and is included here for clarity.}
\label{fig:testsynthesisappendix}
\end{figure}

\subsection{Variable Initialization}
\label{sec:initializationappendix}

We describe primitive variables and reference variables separately. For the case of initializing reference variables, we describe two different strategies:
\begin{itemize}
\item {\bf Null:} Whenever possible, initialize to \texttt{null}.
\item {\bf Instantiation:} Whenever possible, use constructor calls.
\end{itemize}
The first strategy ensures that the unit test does not exhibit additional transfer and alias edges beyond those in the candidate specification. The second strategy may produce a unit test that does not witness correctness, since it may include spurious edges not in the premise of the candidate. However, certain functions require that some of their arguments are not null; for example, the \texttt{put} function in the \texttt{Hashtable} class. We show empirically that the second variant identifies a number of candidates missed by the first, and that these additional specifications are in fact correct.

\paragraph{{\bf\em Primitive initialization.}}

We initialize all primitive variables with \texttt{0} (except characters, which are initialized as \texttt{'a'}). For example, consider the specifications for the \texttt{List} class
\begin{verbatim}
class List { // specifications
  void add(Object ob) { ... }
  Object get(int index) { ... } }
\end{verbatim}
and consider the candidate path specification
\begin{align*}
  s_{\text{list}}=\texttt{ob}\dashrightarrow\texttt{this}_{\texttt{add}}\rightarrow\texttt{this}_{\texttt{get}}\dashrightarrow r_{\texttt{get}},
\end{align*}
which says that if we add an object \texttt{ob} to a list, and then retrieve an object from the same list, then the retrieved object may equal \texttt{ob}. By initializing primitive values to \texttt{0}, the unit test we obtain to check $s_{\text{list}}$ is
\begin{verbatim}
boolean test() { // program
  Object in = new Object(); // o_in
  List list = new List(); // o_list
  list.add(in);
  Object out = list.get(0);
  return in == out; }
\end{verbatim}
which correctly returns true. In our experience, the only important choice of primitive value is the \texttt{index} parameter passed to functions used to retrieve data from collections, e.g., the \texttt{get} method in the \texttt{List} class. Choosing the \texttt{index = 0} retrieves the single object the unit test previously added to the collection. Testing more primitive values is possible but has largely been unnecessary.

\paragraph{{\bf\em Reference initialization using null.}}

Reference variables for which aliasing relations hold must be instantiated (unless they have already been initialized as the return value of a function call). Any other reference variable is initialized to \texttt{null}. For example, in Figure~\ref{fig:testsynthesisappendix}, the variables \texttt{box} and \texttt{out} must be instantiated, but \texttt{boxClone} has already been initialized as the return value of \texttt{clone}. In general, the unit test we synthesize calls the constructor with the fewest number of arguments; primitive arguments are initialized as before, and reference arguments are initialized using \texttt{null}.

\paragraph{{\bf\em Reference initialization using instantiation.}}

In this approach, we have to synthesize constructor calls when empty constructors are unavailable. For example, if the only constructor for the \texttt{Box} class was \texttt{Box(Object val)}, then we would have to initialize an object of type \texttt{Object} as well:
\begin{small}
\begin{verbatim}
Box box = new Box(new Object());
\end{verbatim}
\end{small}

We encode the problem of synthesizing a valid constructor call as a directed hypergraph reachability problem. A \emph{directed hypergraph} is a pair $G=(V,E)$, where $V$ is a set of vertices, and edges $e\in E$ have the form $e=(h,B)$, where $h\in V$ is the \emph{head} of the edge, and $B\subseteq V$ is its \emph{body}. For our purposes, $B$ is a list rather than a set, and may contain a single vertex multiple times.

We construct a hypergraph $G=(V,E)$ where vertices correspond to classes, and edges to constructors:
\begin{itemize}
\item {\bf Vertices:} A vertex $v\in V$ is a library class.
\item {\bf Edges:} An edge $e=(h,B)\in E$ is a constructor, where $h$ is the class of the constructed object and $B$ is the list of classes of the constructor parameters.
\end{itemize}
For convenience, we also include primitive types as vertices in $G$, along with an edge representing the ``empty constructor'', which returns the initialization value described above.

Now, a \emph{path} $T$ in the hypergraph $G=(V,E)$ is a finite tree with root $v_T\in V$ (called the \emph{root} of the path), such that for each vertex $v\in T$, $v$ and its (ordered) children $[v_1,...,v_k]$ are an edge $e_{T,v}=(v,[v_1,...,v_k])\in E$. Note that for each leaf $v$ of $T$, there must necessarily be an edge $(v,[])\in E$, since $v$ has no children. Also, we say a vertex $v\in V$ is \emph{reachable} is there exists a path with root $v$.

In our setting, a path in our hypergraph $G$ corresponds to a call to a constructor---for each vertex $\texttt{X}\in T$ with children $\texttt{X1},...,\texttt{Xk}$, we recursively define the constructor
\begin{align*}
C_T(\texttt{X})=\texttt{new X(}C_T(\texttt{X1})\texttt{, ..., }C_T(\texttt{Xk})\texttt{)}.
\end{align*}
Therefore, devising a constructor call to instantiate an object of type \texttt{X} amounts to computing a path in $G$ with root \texttt{X}. Paths to every reachable vertex can be efficiently computed using a standard dynamic programming algorithm. Furthermore, we can add a weight $w_e$ to each edge in $e\in E$. Then, the \emph{shortest} path (i.e., the path minimizing the total weight $\sum_{v\in T}e_{T,v}$) can similarly be efficiently computed. We choose all weights $w_e=1$ for each $e\in E$.

For example, suppose that the \texttt{Box} class has a single constructor \texttt{Box(Object val)}. Then, our algorithm constructs a hypergraph with two vertices and two edges:
\begin{align*}
V&=\{\texttt{Object},~\texttt{Box}\} \\
E&=\{(\texttt{Object},[]),~(\texttt{Box},[\texttt{Object}])\}.
\end{align*}
Then, the path corresponding to \texttt{Box} is the tree $T=\frac{\texttt{Box}}{\texttt{Object}}$, which corresponds to the constructor call
\begin{small}
\begin{verbatim}
new Box(new Object())
\end{verbatim}
\end{small}
used to instantiate variables of type \texttt{Box}.

As with initializing primitive variables, multiple choices of constructor calls could be used, but selecting a single constructor suffices has been sufficient so far.

\subsection{Statement Scheduling}
\label{sec:scheduling}

Note that the unit test now contains both function call statements as well as variable initialization statements added in the previous step. All the added variable initialization statements can be executed first, so it suffices to schedule the function call statements.

There are two kinds of constraints on scheduling function calls. First, consider edges of the form
\begin{align*}
r_{m_i}\xrightarrow{\Transfer}p_{m_{i+1}}
\end{align*}
in the premise of the candidate path specification; they impose \emph{hard constraints} on the schedule, since $m_i$ must be called before $m_{i+1}$ so its return value can be transfered to $p_{m_{i+1}}$ (edges of the form $p_{m_i}\xrightarrow{\overline{\Transfer}}r_{m_{i+1}}$ impose hard constraints as well). For example, in Figure~\ref{fig:testsynthesisappendix}, the edge $r_{\texttt{clone}}\xrightarrow{\Transfer}\texttt{this}_{\texttt{get}}$ imposes the hard constraint that the call to \texttt{clone} must be scheduled before the call to \texttt{get}. Then, any of the following orderings is permitted:
\begin{align*}
[\texttt{set},\texttt{clone},\texttt{get}],~[\texttt{clone},\texttt{set},\texttt{get}],~[\texttt{clone},\texttt{get},\texttt{set}].
\end{align*}

We use \emph{soft constraints} to choose among schedules satisfying the hard constraints. Empirically, we observe that the order of the functions in the specification is typically the same as the order in which they must be called for the conclusion to be exhibited dynamically. More precisely, function $m_i$ should be called before function $m_j$ whenever $i<j$. In our example, the soft constraint says that \texttt{set} should be scheduled before both \texttt{clone} and \texttt{get}.

Our algorithm iteratively constructs a schedule $[i_1,...,i_k]$ of the function calls $F=\{m_1,...,m_k\}$. At iteration $t$, it selects the $t$th function call $m_{i_t}$ from the remaining calls $F_t\subseteq F$. It does so greedily, by identifying the choices $G_t\subseteq F_t$ that satisfy the hard constraints, and then selecting $m_{i_t}\in G_t$ to be optimal according to the soft constraints. These conditions uniquely specify $m_{i_t}$, since our soft constraints are a total ordering.

Our algorithm keeps track of the remaining statements $F_t$ as a directed acyclic graph (DAG), which includes an edge $m_i\to m_j$ for each hard constraint that $m_i$ should be scheduled before $m_j$. Then, $G_t$ is the set of roots of $F_t$. Furthermore, our algorithm maintains $G_t$ as a priority queue, where the priority of $m_i$ is $i$ (the highest priority element in $G_t$ is the element with the smallest index $i$).

We initialize $F_1=F$; then, $G_1$ is the subset of vertices in $F_1$ without a parent. Updates are computed as follows:
\begin{enumerate}
\item The highest priority function call $m_{i_t}$ in $G_t$ is removed from both $G_t$ and from $F_t$.
\item For each child $m_i$ of $m_{i_t}$ in $F_t$, we determine if $m_i$ is now a root of $F_t$ (i.e., none of its parents are in $F_t$).
\item For every child $m_i$ that is now a root of $F_t$, we add $m_i$ to $G_t$ with priority $i$.
\end{enumerate}

In Figure~\ref{fig:testsynthesisappendix}, $F_1$ has three vertices \texttt{set} (priority $1$), \texttt{clone} (priority $2$), and \texttt{get} (priority $3$), and a single edge $\texttt{clone}\to\texttt{get}$, and $G_1$ includes \texttt{set} and \texttt{clone}. Therefore, the selected schedule is $[\texttt{set},\texttt{clone},\texttt{get}]$.

\section{Proof of Equivalence Theorem}
\label{sec:mainproof}

We prove Theorem~\ref{THM:EQUIV}, relegating the proof of technical lemmas to Appendix~\ref{sec:lemproof}. To simplify the proof, we assume the following:
\begin{assumption}
\label{assump:disjoint}
\rm
Let $\mathcal{F}_{\text{lib}}$ be fields accessed by the library and $\mathcal{F}_{\text{prog}}$ be fields accessed by the program, and let the \emph{shared fields} be $\mathcal{F}_{\text{share}}=\mathcal{F}_{\text{lib}}\cap\mathcal{F}_{\text{prog}}$. We assume $\mathcal{F}_{\text{share}}=\varnothing$.
\end{assumption}
We can remove this assumption by having the static analysis treat accesses to library fields in the program as calls to getter and setter library functions.

\subsection{Converting the Library Implementation to Path Specifications}
\label{sec:conversion}

First, we describe how to convert the library implementation into a set $S$ of transfer and proxy object specifications. A specification of the form
\begin{align*}
z_1\dashrightarrow w_1\to...\to z_k\dashrightarrow w_k.
\end{align*}
is included in $S$ if there exist paths
\begin{align*}
z_1\xdashrightarrow{\beta_1}w_1,\sss...,\sss z_k\xdashrightarrow{\beta_k}w_k
\end{align*}
such that $A\xRightarrow{*}\beta_1\tilde{\alpha}_1...\tilde{\alpha}_{k-1}\beta_k$ in $C_{\text{pt}}$, where
\begin{align*}
A=
\begin{cases}
\Transfer&\text{if }z_1=p_{m_1} \\
\Alias&\text{if }z_1=r_{m_1}
\end{cases}
\end{align*}
and
\begin{align*}
\tilde{\alpha}_i=
\begin{cases}
\Assign&\text{if }w_i=p_{m_i}\text{ and }z_{i+1}=r_{m_{i+1}} \\
\overline{\Assign}&\text{if }w_i=r_{m_i}\text{ and }z_{i+1}=p_{m_{i+1}} \\
\overline{\New}~\New&\text{if }w_i=p_{m_i}\text{ and }z_{i+1}=p_{m_{i+1}}.
\end{cases}
\end{align*}
Then, we prove that the conclusion of Theorem~\ref{THM:EQUIV} holds for $S$ constructed with this algorithm.

\subsection{Proof Overview}

Let $\overline{G}$ denote the points-to sets computed by running the static analysis with the library implementation available, and $\overline{G}(S)$ denote the points-to sets computed by running the static analysis with the path specifications $S$. We have to prove that $\overline{G}=\overline{G}(S)$; the direction $\overline{G}(S)\subseteq\overline{G}$ follows easily, since a path specification $s$ is included in $S$ exactly when the library implementation would imply the same logical formula as the semantics of $s$.

The challenging direction is to show that $S$ is sound, i.e.,
\begin{align*}
\overline{G}\subseteq\overline{G}(S).
\end{align*}
For simplicity, we focus on points-to edges $o\xrightarrow{\FlowsTo}x$; the alias and transfer relations follow similarly. Suppose that $o\xrightarrow{\FlowsTo}y\in\overline{G}(S)$; then, there must exist a path $o\xrightarrow{\New}x\xdashrightarrow{\alpha}y$, where $\Transfer\xRightarrow{*}\alpha$. This path passes into and out of library functions, leading to a decomposition
\begin{align}
\label{eqn:path}
x\xdashrightarrow{\alpha_0}z_1\xdashrightarrow{\beta_1}w_1\xdashrightarrow{\alpha_1}...\xdashrightarrow{\beta_k}w_k\xdashrightarrow{\alpha_k}y,
\end{align}
where $\alpha=\alpha_0\beta_1\alpha_1...\beta_k\alpha_k$. This decomposition suggests that the following path specification may be applied to derive $x\xrightarrow{\Transfer}y$:
\begin{align}
\label{eqn:spec}
z_1\dashrightarrow w_1\to...\to z_k\dashrightarrow w_k.
\end{align}

At a high level, our proof has two parts. First, we prove the case where the segments of $\alpha$ in the program do not contain field accesses, i.e., $\alpha\in(\Sigma_{\text{free}}\cup\Sigma_{\text{lib}})^*$, where
\begin{align*}
\Sigma_{\text{free}}&=\{\Assign,\overline{\Assign},\New,\overline{\New}\} \\
\Sigma_{\text{prog}}&=\{\Store[f],\Load[f],\overline{\Store[f]},\overline{\Load[f]}\mid f\in\mathcal{F}_{\text{prog}}\} \\
\Sigma_{\text{lib}}&=\{\Store[f],\Load[f],\overline{\Store[f]},\overline{\Load[f]}\mid f\in\mathcal{F}_{\text{lib}}\}.
\end{align*}
Second, we show how ``nesting'' of fields allows us to reduce the general case to the case $\alpha\in(\Sigma_{\text{free}}\cup\Sigma_{\text{lib}}\cup\Sigma_{\text{prog}})^*$. In particular, by Assumption~\ref{assump:disjoint}, the library field accesses and program field accesses do not match one another. As previously discussed, this assumption can be enforced by a purely syntactic program transformation where accesses to library fields in the program are converted into calls to getter and setter functions.

Consider a path of the form (\ref{eqn:path}) such that $\alpha\in(\Sigma_{\text{free}}\cup\Sigma_{\text{lib}})^*$. We need to show that in this case, we derive the edge $x\xrightarrow{\Transfer}y\in\overline{G}(S)$, where $S$ is constructed as in Section~\ref{sec:conversion}. Our proof of this claim relies on two results. The first result says that for such a path, the conclusion of (\ref{eqn:spec}) holds when each $w_i$ is connected to $z_{i+1}$ by $\alpha_i$:
\begin{proposition}
\label{prop:conclusion}
\rm
For any path of the form (\ref{eqn:path}) such that $\alpha\in(\Sigma_{\text{free}}\cup\Sigma_{\text{lib}})^*$ we have (i) the case $w_i=r_i$ and $z_{i+1}=r_{i+1}$ cannot happen, and (ii) $\Transfer\xRightarrow{*}\beta_1\alpha_1\beta_2...\alpha_{k-1}\beta_k$.
\end{proposition}
As a consequence of this result, we know that the path specification (\ref{eqn:spec}) is contained in $S$. The second result says that the premise of (\ref{eqn:spec}) holds for our case:
\begin{proposition}
\label{prop:premise}
\rm
For any path of the form (\ref{eqn:spec}) such that $\alpha\in(\Sigma_{\text{free}}\cup\Sigma_{\text{lib}})^*$, we have
\begin{alignat*}{2}
&A_i\xRightarrow{*}\alpha_i\hspace{0.2in}&&(\forall i\in[k-1]) \\
&A_i\xRightarrow{*}\alpha_i\hspace{0.2in}&&(\forall i\in\{0,k\}).
\end{alignat*}
\end{proposition}
Therefore, we can conclude that when running the static analysis using path specifications, we derive the conclusion of the path specification (\ref{eqn:spec}), i.e., $z_1\xrightarrow{\Transfer}w_k\in\overline{G}(S)$. In summary, we have the following result:
\begin{theorem}
\label{THM:EQUIVfree}
\rm
Theorem~\ref{THM:EQUIV} holds for any $\alpha\in(\Sigma_{\text{free}}\cup\Sigma_{\text{lib}})^*$.
\end{theorem}
\begin{proof}
Consider an edge $x\xrightarrow{\Transfer}y\in\overline{G}$ derived by the static analysis using the library implementation. We claim that this edge is derived by the static analysis when using path specifications, i.e., $x\xrightarrow{\Transfer}y\in\overline{G}(S)$. By Proposition~\ref{prop:conclusion}, we conclude that (\ref{eqn:spec}) is in $S$. Furthermore, by Proposition~\ref{prop:premise}, the premise of (\ref{eqn:spec}) holds, so the static analysis derives its conclusion, i.e., $z_1\xrightarrow{\Transfer}w_k\in\overline{G}(S)$. Therefore, we have
\begin{align*}
x\xrightarrow{\Transfer}z_1\xrightarrow{\Transfer}w_k\xrightarrow{\Transfer}y\in\overline{G}(S),
\end{align*}
so the static analysis derives $x\xrightarrow{\Transfer}y\in\overline{G}(S)$, as claimed.

Now, we know that any points-to edge $o\xrightarrow{\FlowsTo}y\in\overline{G}$ has the form $o\xrightarrow{\New}x\xrightarrow{\Transfer}y$. Since we have shown that $x\xrightarrow{\Transfer}y\in\overline{G}(S)$, the static analysis also derives $o\xrightarrow{\FlowsTo}y\in\overline{G}(S)$, so the result follows.
\end{proof}

In the remainder of the section, we introduce the technical machinery that enables us to reason about ``equivalence'' of the semantics of different sequences of statements. Then, we describe how we prove Propositions~\ref{prop:conclusion} \&~\ref{prop:premise}. Finally, we reduce Theorem~\ref{THM:EQUIV} to Theorem~\ref{THM:EQUIVfree}.

\subsection{Equivalent Semantics}

Proving Propositions~\ref{prop:conclusion} \&~\ref{prop:premise} requires reasoning about the \emph{equivalence} of the semantics of sequences of statements in $P$. For example, to prove Proposition~\ref{prop:conclusion}, we show that each $\alpha_i$ is ``equivalent'' to $\tilde{\alpha}_i$. Intuitively, for $\tilde{\alpha}_i=\Assign$, we show that the sequence of statements represented by $\alpha_i$ exhibits the same semantics as a single assignment. For example, $y\gets x,~z\gets y$ has the same points-to effects as $z\gets x$ (assuming $y$ is temporary). We leverage the correspondence established by formulating points-to analysis as context-free language reachability:
\begin{align*}
\text{sequence of statements}~=~\text{sequence $\alpha\in\Sigma^*$}.
\end{align*}
For example, the first sequence of statements above corresponds to $(\Assign~\Assign)$, and the second to $\Assign$.

Using this correspondence, we can reduce reasoning about sequences of statements with equivalent semantics to studying equivalence classes of strings $\alpha\in\Sigma^*$:
\begin{align*}
  \begin{array}{c}
    \text{equivalent sequences} \\
    \text{of statements}
  \end{array}
  ~=~
  \begin{array}{c}
    \text{equivalence classes} \\
         {[\alpha]}\subseteq\Sigma^*
  \end{array}.
\end{align*}
In particular, $\alpha,\beta\in\Sigma^*$ are \emph{equivalent} if
\begin{align}
\gamma\alpha\delta\in\mathcal{L}(C_{\text{pt}})\Leftrightarrow\gamma\beta\delta\in\mathcal{L}(C_{\text{pt}})\sss(\forall\gamma,\delta\in\Sigma^*).
\end{align}
In other words, $\alpha$ can be used interchangeably with $\beta$ in any string without affecting whether the string is contained in $\mathcal{L}(C_{\text{pt}})$. We use $[\alpha]=\{\beta\in\Sigma^*\mid\alpha\sim\beta\}$ to denote the equivalence class of $\alpha\in\Sigma^*$. Then, $[\alpha]=[\beta]$ if for any two paths
\begin{align*}
o\xdashrightarrow{\gamma}v\xdashrightarrow{\alpha}w\xdashrightarrow{\delta}x,\hspace{0.2in}o\xdashrightarrow{\gamma}v\xdashrightarrow{\beta}w\xdashrightarrow{\delta}x,
\end{align*}
the first results in $x\hookrightarrow o$ if and only if the second does. For example, $[\Assign~\Assign]=[\Assign]$.

Then, equivalence is compatible with sequencing:
\begin{lemma}
\label{lem:semigroup}
\rm
If $[\alpha]=[\alpha']$ and $[\beta]=[\beta']$, then $[\alpha\beta]=[\alpha'\beta']$.
\end{lemma}
\begin{proof}
By definition, $\gamma\alpha\beta\delta\Leftrightarrow\gamma\alpha'\beta\delta\Leftrightarrow\gamma\alpha'\beta'\delta$.
\end{proof}

In particular, Lemma~\ref{lem:semigroup} shows that sequencing is well-defined for equivalence classes:
\begin{align}
\label{eqn:quotient}
[\alpha]~[\beta]=[\alpha\beta],
\end{align}
since different choices $\alpha'\in[\alpha]$ and $\beta'\in[\beta]$ yield the same equivalence class, i.e., $[\alpha\beta]=[\alpha'\beta']$. Abstractly, $\Sigma^*$ is a semigroup, with sequencing as the semigroup operation; then, Lemma~\ref{lem:semigroup} shows the equivalence relation is compatible with the semigroup operation, so the quotient $\Sigma/\sim$ is a semigroup with semigroup operation (\ref{eqn:quotient}).

For convenience, we let $\phi$ denote an element of the equivalence class of strings such that
\begin{align}
\text{for all }\gamma,\delta\in\Sigma^*,\sss\gamma\phi\delta\not\in\mathcal{L}(C_{\text{pt}}).
\end{align}
In other words, $[\phi]$ describes sequences of statements that can never be completed to a valid flows-to path.

\subsection{Proofs of Propositions~\ref{prop:conclusion} \&~\ref{prop:premise}}

Now, we describe how to prove that under the conditions of Proposition~\ref{prop:conclusion}, $[\alpha_i]=[\tilde{\alpha}_i]$, which suffices to prove the proposition. We focus on the case $\tilde{\alpha}_i=\Assign$; the other cases are similar. We need the following technical lemma (we give a proof in Appendix~\ref{sec:assignproof}):
\begin{lemma}
\label{lem:assign}
\rm
For any $\alpha\in\Sigma_{\text{free}}^*$, we have
\begin{align*}
[\Assign]~[\alpha]~[\Assign]\in\{[\Assign],~[\phi]\}.
\end{align*}
\end{lemma}
With this lemma, since $\tilde{\alpha}_i=\Assign$, $w_{m_i}=r_{m_i}$ and $z_{m_{i+1}}=p_{m_i}$, so the path $w_{m_i}\xdashrightarrow{\alpha_i}z_{m_{i+1}}$ has form
\begin{align*}
w_{m_i}=r_{m_i}\xrightarrow{\Assign}y_i\xdashrightarrow{\alpha_i'}x_{i+1}\xrightarrow{\Assign}p_{m_{i+1}}=z_{m_{i+1}},
\end{align*}
where $\alpha_i=\Assign~\alpha_i'~\Assign$. By Lemma~\ref{lem:assign},
\begin{align*}
[\alpha_i]=[\Assign]~[\alpha_i']~[\Assign]\in\{[\Assign],~[\phi]\}.
\end{align*}
Since $(\New~\alpha)\in\mathcal{L}(C_{\text{pt}})$, we cannot have $[\alpha_i]=[\phi]$, so
\begin{align*}
[\alpha_i]=[\Assign]=[\tilde{\alpha}_i],
\end{align*}
as claimed. We have also proven the claim in Proposition~\ref{prop:premise} that $A_i\xRightarrow{*}\alpha_i$ (with $A_i=\Transfer$) also follows. The other claims in Propositions~\ref{prop:conclusion} \&~\ref{prop:premise} follow similarly.~$\square$

\subsection{Reduction of Theorem~\ref{THM:EQUIV} to Theorem~\ref{THM:EQUIVfree}}

To handle field accesses, we use the fact that pairs of terminals $(\Store[f],\Load[f])$ and $(\overline{\Load[f]},\overline{\Store[f]})$ in strings $\alpha\in\mathcal{L}(C_{\text{pt}})$ are matching. Therefore, can identify an inner-most nested pair $(\sigma,\tau)$ such that the string $\beta$ between $\sigma$ and $\tau$ contains no field accesses, i.e., $\beta\in\Sigma_{\text{free}}$. Furthermore, by Assumption~\ref{assump:disjoint}, library field accesses and program field accesses do not match one another. In particular, the set of matching program field accesses is
\begin{align*}
\Delta_{\text{prog}}=\bigcup_{f\in\mathcal{F}_{\text{prog}}}\{(\Store[f],\Load[f]),~(\overline{\Load[f]},\overline{\Store[f]})\}.
\end{align*}
\begin{lemma}
\label{lem:min}
\rm
For any $\alpha\in\mathcal{L}(C_{\text{pt}})$, either $\alpha\in(\Sigma_{\text{free}}\cup\Sigma_{\text{lib}})^*$, or there exists a pair of terminals $(\sigma,\tau)\in\Delta_{\text{prog}}$ such that $\alpha=\gamma\sigma\beta\tau\delta$, where $\gamma,\delta\in\Sigma^*$ and $\beta\in(\Sigma_{\text{free}}\cup\Sigma_{\text{lib}})^*$.
\end{lemma}
The next step is to characterize $[\sigma\beta\tau]$:
\begin{lemma}
\label{lem:match}
\rm
For any $(\sigma,\tau)\in\Delta_{\text{prog}}$ and $\beta\in(\Sigma_{\text{free}}\cup\Sigma_{\text{lib}})^*$,
\begin{align*}
[\sigma]~[\beta]~[\tau]\in\{[\Assign],~[\phi]\}.
\end{align*}
\end{lemma}
Finally, $\beta$ must be an aliasing relation:
\begin{lemma}
\label{lem:field}
\rm
For any $\beta\in\Sigma^*$,
\begin{align*}
[\Store[f]]~[\beta]~[\Load[f]]=[\Assign]&\Rightarrow[\beta]=[\overline{\New}~\New] \\
[\overline{\Load[f]}]~[\beta]~[\overline{\Store[f]}]=[\Assign]&\Rightarrow[\beta]=[\overline{\New}~\New].
\end{align*}
\end{lemma}
Now, if $\alpha\in\Sigma_{\text{free}}^*$, we are done. Otherwise, putting the three lemmas together, we perform the following procedure:
\begin{enumerate}
\item By Lemma~\ref{lem:min}, we can write $\alpha=\gamma\sigma\beta\tau\delta$, where $(\sigma,\tau)\in\Delta_{\text{prog}}$ and $\beta\in(\Sigma_{\text{free}}\cup\Sigma_{\text{lib}})^*$, such that
\begin{align*}
y\xdashrightarrow{\gamma}v\xrightarrow{\sigma}w\xdashrightarrow{\beta}t\xrightarrow{\tau}u\xdashrightarrow{\delta}x.
\end{align*}
\item By Lemma~\ref{lem:match}, $[\sigma]~[\beta]~[\tau]=[\Assign]$.
\item By Lemma~\ref{lem:field}, $[\beta]=[\overline{\New}~\New]$.
\item By Theorem~\ref{THM:EQUIV}, we have $w\xrightarrow{\Alias}t\in\overline{G}(\tilde{S})$; therefore, $v\xdashrightarrow{\Transfer}u\in\overline{G}(\tilde{S})$ as well.
\item Recursively apply the procedure to $\alpha'=\gamma~\Assign~\delta$.
\end{enumerate}
This procedure must terminate, since $\alpha$ has finitely many pairs of store and load statements. Theorem~\ref{THM:EQUIV} follows.~$\square$

\section{Proof of Technical Lemmas}
\label{sec:lemproof}

We prove the technical lemmas used in Appendix~\ref{sec:mainproof}.

\subsection{Proof of Lemma~\ref{lem:assign}}
\label{sec:assignproof}

We first show the following lemma, which completely characterizes the subgroupoid of elements $\Sigma_{\text{free}}^*\subseteq\Sigma^*$:
\begin{lemma}
\label{lem:multtable}
\rm
We have
\begin{align*}
[\Assign]~[\Assign]&=[\Assign] \\
[\Assign]~[\overline{\Assign}]&=[\phi] \\
[\overline{\Assign}]~[\Assign]&=[\phi] \\
[\overline{\Assign}]~[\overline{\Assign}]&=[\overline{\Assign}] \\
[\Assign]~[\overline{\New}~\New]&=[\phi] \\
[\overline{\New}~\New]~[\Assign]&=[\overline{\New}~\New] \\
[\overline{\Assign}]~[\overline{\New}~\New]&=[\overline{\Assign}] \\
[\overline{\New}~\New]~[\overline{\Assign}]&=[\phi] \\
[\overline{\New}~\New]~[\overline{\New}~\New]&=[\phi].
\end{align*}
\end{lemma}
\begin{proof}
We show the first relation; the others follow similarly. First, we show that if $\gamma~\Assign~\delta\in\mathcal{L}(C_{\text{pt}})$, then $\gamma~\Assign~\Assign\delta\in\mathcal{L}(C_{\text{pt}})$. There must exist a derivation
\begin{align*}
\FlowsTo
\Rightarrow...
&\Rightarrow\Transfer~u_{\delta} \\
&\Rightarrow\Transfer~\Assign~u_{\delta} \\
&\Rightarrow... \\
&\Rightarrow\gamma~\Assign~\delta
\end{align*}
since the only production in $C_{\text{pt}}$ containing the terminal symbol $\Assign$ is $\Transfer\to\Transfer~\Assign$. Therefore, the following derivation also exists:
\begin{align*}
\FlowsTo
\Rightarrow...
&\Rightarrow\Transfer~\delta \\
&\Rightarrow\Transfer~\Assign~u_{\delta} \\
&\Rightarrow\Transfer~\Assign~\Assign~u_{\delta} \\
&\Rightarrow... \\
&\Rightarrow\gamma~\Assign~\Assign~\delta,
\end{align*}
i.e., $\gamma~\Assign~\Assign~\delta\in\mathcal{L}(C_{\text{pt}})$. By a similar argument, it follows that if $\gamma~\Assign~\Assign~\delta\in\mathcal{L}(C_{\text{pt}})$, then $\gamma~\Assign~\delta\in\mathcal{L}(C_{\text{pt}})$, so $[\Assign]~[\Assign]=[\Assign]$.
\end{proof}
\vspace{5pt}

\noindent It follows directly that if $\alpha\in\Sigma_{\text{free}}^*$, then
\begin{align*}
[\alpha]\in\{[\phi],~[\epsilon],~[\Assign],~[\overline{\Assign}],~[\overline{\New}~\New]\}.
\end{align*}
%
%
%
%
%
In particular, for $\alpha'\in\Sigma_{\text{free}}^*$, $[\Assign]~[\alpha']\in\{[\Assign],~[\phi]\}$, so the lemma follows by taking $\alpha'=\alpha~\Assign$.~$\square$

\subsection{Proof of Lemma~\ref{lem:min}}

If we replace the terminal symbols $\sigma\in\Sigma_{\text{free}}$ with $\epsilon$ in $C_{\text{pt}}$, then $C_{\text{pt}}$ is a parentheses matching grammar where each ``open parentheses'' $\Store[f]$ (resp., $\overline{\Load[f]}$) must be matched with a corresponding ``closed parentheses'' $\Load[f]$ (resp., $\overline{\Store[f]}$). Also, by Assumption~\ref{assump:disjoint}, $\Sigma_{\text{lib}}\cap\Sigma_{\text{prog}}=\varnothing$.

Now, we prove by induction on the length of $\alpha$. The base case $\alpha=\epsilon$ is clear. If $\alpha\in\Sigma^*$ does not contain a pair of matched parentheses $(\Store[f],\Load[f])\in\Sigma_{\text{prog}}^2$, then $\alpha\in(\Sigma_{\text{free}}\cup\Sigma_{\text{lib}})^*$, so we are done. Otherwise, for any such pair of matched parentheses, we can express $\alpha=\gamma\sigma\alpha'\tau\delta$. By induction, the lemma holds for $\alpha'$, so we can write $\alpha=\gamma'\sigma'\beta'\tau'\delta'$ as in the lemma. Therefore, we have
\begin{align*}
\alpha=(\gamma\sigma\gamma')\sigma'\beta'(\tau'\delta'\tau\delta),
\end{align*}
so the claim follows.~$\square$

\subsection{Proof of Lemma~\ref{lem:match}}

We show the case $(\sigma,\tau)=(\Store[f],\Load[f])$, where $f\in\mathcal{F}_{\text{prog}}$; the case $(\sigma,\tau)=(\overline{\Load[f]},\overline{\Store[f]})$ is similar. First, suppose that $\gamma\sigma\beta\tau\delta\in\mathcal{L}(C_{\text{pt}})$. Then, there must exist a derivation of form
\begin{align*}
\FlowsTo
\Rightarrow...
&\Rightarrow u_{\gamma}~\overline{\Transfer}~u_{\delta} \\
&\Rightarrow u_{\gamma}~\Transfer~\sigma~\Alias~\tau~u_{\delta} \\
&\Rightarrow... \\
&\Rightarrow\gamma\sigma\beta\tau\delta,
\end{align*}
so the following derivation exists:
\begin{align*}
\FlowsTo
\Rightarrow...
\Rightarrow u_{\gamma}~\overline{\Transfer}~u_{\delta} \\
&\Rightarrow u_{\gamma}~\Transfer~\Assign~u_{\delta} \\
&\Rightarrow... \\
&\Rightarrow\gamma~\Assign~\delta.
\end{align*}
The converse follows similarly, so the claim follows.~$\square$

\subsection{Proof of Lemma~\ref{lem:field}}

We show two preliminary lemmas.
\begin{lemma}
\label{lem:p1}
\rm
We have
\begin{align*}
[\Store[f]]~[\overline{\New}~\New]~[\Load[f]]&=[\Assign] \\
[\overline{\Load[f]}]~[\overline{\New}~\New]~[\Store[f]]&=[\overline{\Assign}].
\end{align*}
\end{lemma}
\begin{proof}
Suppose that $\gamma~\Store[f]~\overline{\New}~\New~\Load[f]~\delta\in\mathcal{L}(C_{\text{pt}})$. Then, we must have derivation
\begin{align*}
\FlowsTo
\Rightarrow...
&\Rightarrow u_{\gamma}~\Transfer~u_{\delta} \\
&\Rightarrow u_{\gamma}~\Store[f]~\Alias~\Load[f]~u_{\delta} \\
&\Rightarrow... \\
&\Rightarrow\gamma~\Store[f]~\alpha~\Load[f]~\delta,
\end{align*}
so we also have derivation
\begin{align*}
\FlowsTo
\Rightarrow...
&\Rightarrow u_{\gamma}~\Transfer~u_{\delta} \\
&\Rightarrow u_{\gamma}~\Assign~u_{\delta} \\
&\Rightarrow... \\
&\Rightarrow\gamma~\Assign~\Load[f]~\delta.
\end{align*}
Thus, $\gamma~\Assign~\delta\in\mathcal{L}(C_{\text{pt}})$. The converse follows similarly, as does the second claim.
\end{proof}

\begin{lemma}
\label{lem:p2}
\rm
For any $\beta\in\Sigma^*\setminus\{\epsilon\}$, we have
\begin{align*}
[\beta]=[\Assign]&\Leftrightarrow\beta\in\mathcal{L}(C_{\text{pt}},\Transfer) \\
[\beta]=[\overline{\Assign}]&\Leftrightarrow\beta\in\mathcal{L}(C_{\text{pt}},\overline{\Transfer}) \\
[\beta]=[\overline{\New}~\New]&\Leftrightarrow\beta\in\mathcal{L}(C_{\text{pt}},\Alias).
\end{align*}
\end{lemma}
\begin{proof}
We first show the forward implication. If $[\beta]=[\Assign]$, then $\New~\Assign\in\mathcal{L}(C_{\text{pt}})$, so $\New~\beta\in\mathcal{L}(C_{\text{pt}})$. Therefore, there must exist a derivation
\begin{align*}
\FlowsTo\Rightarrow\New~\Transfer~\Rightarrow...~\Rightarrow\New~\beta,
\end{align*}
so $\beta\in\mathcal{L}(C_{\text{pt}},\Transfer)$. The other two cases follow similarly. Now, we show the backward implication. Suppose that $\beta\in\mathcal{L}(C_{\text{pt}},\Transfer)$. We prove by structural induction on the derivation of $\beta$ from $\Transfer$. Since $\beta\not=\epsilon$, $\beta$ cannot have been produced by $\Transfer\Rightarrow\epsilon$. If $\beta$ is produced by $\Transfer\rightarrow\Transfer~\Assign$, then $\beta=\beta'\Assign$, where $\beta'\in\mathcal{L}(C_{\text{pt}},\Transfer)$. By induction, $[\beta']=[\Assign]$, so
\begin{align*}
[\beta]=[\beta']~[\Assign]=[\Assign]~[\Assign]=[\Assign],
\end{align*}
where the last step follows from Lemma~\ref{lem:multtable}. Next, if $\beta$ is produced using the production
\begin{align*}
\Transfer\rightarrow\Transfer~\Store[f]~\Alias~\Load[f],
\end{align*}
then $\beta=\beta'~\Store[f]~\beta''~\Load[f]$, where $\beta'\in\mathcal{L}(C_{\text{pt}},\Transfer)$ and $\beta''\in\mathcal{L}(C_{\text{pt}},\Alias)$. By induction, $[\beta']=[\Assign]$ and $[\beta'']=[\overline{\New}~\New]$, so
\begin{align*}
\beta
&=[\beta']~[\Store[f]]~[\beta'']~[\Load[f]] \\
&=[\Assign]~[\Store[f]]~[\overline{\New}~\New]~[\Load[f]] \\
&=[\Assign],
\end{align*}
where the last step follows from Lemma~\ref{lem:p1} and Lemma~\ref{lem:multtable}. The remaining cases follow similarly.
\end{proof}
\vspace{5pt}

\noindent Now, suppose that $[\Store[f]]~[\beta]~[\Load[f]]=[\Assign]$. Since
\begin{align*}
\New~\Store[f]~\overline{\New}~\New~\Load[f]\in\mathcal{L}(C_{\text{pt}}),
\end{align*}
we have
\begin{align*}
\New~\Store[f]~\beta~\Load[f]\in\mathcal{L}(C_{\text{pt}}),
\end{align*}
so the following derivation must exist:
\begin{align*}
\FlowsTo&\Rightarrow\New~\Store[f]~\Alias~\Load[f] \\
&\Rightarrow... \\
&\Rightarrow\New~\Store[f]~\beta~\Load[f],
\end{align*}
i.e., $\beta\in\mathcal{L}(C_{\text{pt}},\Alias)$. Finally, by Lemma~\ref{lem:p2}, we have $[\beta]=[\overline{\New}~\New]$. The second case follows similarly.~$\square$

\section{Proof of Correctness of Unit Test Synthesis}
\label{sec:witnessproof}

In this section, we sketch a proof of Theorem~\ref{THM:TESTSYNTHESIS}, which says that the unit tests synthesized by our algorithm are potential witnesses.

\subsection{General Condition}

First, we establish a general condition for $P$ to be a potential witness:
\begin{proposition}
\label{prop:witness}
\rm
Let $s$ be a path specification with premise $(e_1\in\overline{G})\wedge...\wedge(e_k\in\overline{G})$. A program $P$ is a potential witness of $s$ if the set of edges $\{e_1,...,e_k\}$ in the premise of $s$ exactly equals
\begin{align*}
  \left\{
  w\xrightarrow{A}z\in\overline{G}(P,\varnothing)\bigm\vert
  \begin{array}{l}
    w,z\in\mathcal{V}_{\text{lib}}\text{ and }\\
    A\in\{\Transfer,\overline{\Transfer},\Alias\}
  \end{array}
  \right\}.
\end{align*}
\end{proposition}
\begin{proof}
Let $P$ be a potential witness for $s$, and suppose that the conclusion of $s$ is $(e\in\overline{G})$. Let $S$ be a set of path specifications that computes $e$ for $P$, i.e., $e\in\overline{G}(P,S)$. We need to show that for any such $S$, $S\cup\{s\}$ is equivalent to $S$. Clearly, $S\cup\{s\}$ has higher or equal recall than $S$, so it suffices to show that it also has higher or equal precision than $S$. Consider an arbitrary program $P'$. Then, if $s$ is used during the computation $\overline{G}(P,S\cup\{s\})$, then at that point, the premise of $s$ holds for $\overline{G}$, i.e., $e_1,...,e_k\in\overline{G}$. Since the graph for $P$ is contained in the graph for $P'$, and our static analysis is monotone, we have $e\in\overline{G}(P,S)\subseteq\overline{G}(P',S)$, i.e., $e$ is computed without $s$. Thus, $\overline{G}(P',S\cup\{s\})=\overline{G}(P',S)$, so $S\cup\{s\}$ equivalent to $S$ as claimed.
\end{proof}

\subsection{Proof Sketch of Theorem~\ref{THM:TESTSYNTHESIS}}

Let $s=z_1\to w_1\dashrightarrow...\to z_k\dashrightarrow w_k$. Since the function calls are treated as no-ops by the static analysis (according to the definition of a potential witness), they do not add any edges to the extracted graph $G$ except for assignments to and from parameters and return values. The only other edges in the graph $G$ extracted from $P$ are those corresponding to the allocation statements added to $P$ in the initialization step.

First, we show that the edges in the premise of $s$ are contained in $\overline{G}(P,\varnothing)$. For an edge $w_i\to z_{i+1}$, there are three possibilities---either $A_i=\Transfer$, $A_i=\overline{\Transfer}$, or $A_i=\Alias$:
\begin{itemize}
\item {\bf Case $A_i=\Transfer$:} Then, $w_i$ is a return value and $z_{i+1}$ is a parameter. Then, the unit test synthesis algorithm assigns the return value of $m_i$ to the argument of $m_{i+1}$, i.e., the edges
\begin{align*}
w_i\xrightarrow{\Assign}x\xrightarrow{\Assign}z_{i+1}\in G,
\end{align*}
where $G$ is the graph extracted from $P$. Therefore, we have $(w_i\xrightarrow{\Transfer}z_{i+1})\in\overline{G}(P,\varnothing)$.
\item {\bf Case $A_i=\overline{\Transfer}$:} This case is analogous to the case $A=\Transfer$.
\item {\bf Case $A_i=\Alias$:} Then, $w_i$ and $z_{i+1}$ are both parameters. Then, $w_i$ and $z_{i+1}$ are both parameters. Then, the unit test synthesis algorithm allocates a new object and passes it as a parameter to each $m_i$ and $m_{i+1}$, i.e., the edges
\begin{align*}
o\xrightarrow{\New}x\xrightarrow{\Assign}w_i\in G\text{ and }o\xrightarrow{\New}x\xrightarrow{\Assign}z_{i+1}\in G.
\end{align*}
Therefore, we have $(w_i\xrightarrow{\Alias}z_{i+1})\in\overline{G}(P,\varnothing)$.
\end{itemize}

Second, consider all edges $w\xrightarrow{A_i}z$, where $w,z\in\mathcal{V}_{\text{lib}}$ and $A_i\in\{\Transfer,\overline{\Transfer},\Alias\}$, that are contained in the premise of $s$. By inspection, of the edges in $G$ as described above, the only additional edges in $\overline{G}(P,\varnothing)$ of this form are:
\begin{itemize}
\item The self-loops $z_i\xrightarrow{\Transfer}z_i$ and $w_i\xrightarrow{\Transfer}w_i$ (since there is a production $\Transfer\to\epsilon$ in the points-to grammar $C_{\text{pt}}$).
\item The backward edges $z_{i+1}\xrightarrow{\overline{A_i}}w_i$, where we have $A_i\in\{\Transfer,\overline{\Transfer}\}$).
\end{itemize}
If these edges were added to the premise of $s$ for $P$, then by Proposition~\ref{prop:witness}, we could conclude that $P$ is a potential witness of $s$. However, these edges are in $\overline{G}(P,S)$ for any program $P$ and any specifications $S$, so we can add them to the premise of $s$ without affecting its semantics. It follows that if $P$ is a witness for $s'$, and $s'$ is equivalent to $s$, then $P$ is a witness for $s$ as well. Therefore, $P$ is a witness for $s$ as claimed.~$\square$

\end{document}